\documentclass[11pt]{article}
\usepackage{amsmath,amssymb,amsfonts,amsthm,mathtools}
\usepackage{subcaption}
\usepackage{graphicx}
\usepackage{fullpage}
\usepackage[hyphens]{url}
\usepackage[backref=page]{hyperref}
\usepackage[hyphenbreaks]{breakurl}
\usepackage{color}
\usepackage{wrapfig}
\usepackage{tikz}
\usetikzlibrary{decorations.pathreplacing}
\usepackage{setspace}
\usepackage{algorithm}
\usepackage[noend]{algpseudocode}
\usepackage[framemethod=tikz]{mdframed}
\usepackage{xspace}
\usepackage{pgfplots}
\usepackage{framed}
\usepackage{subcaption}
\usepackage{thmtools}
\usepackage{thm-restate}
\usepackage{comment}
\usepackage{bm}
\pgfplotsset{compat=1.5}

\newtheorem{theorem}{Theorem}[section]
\newtheorem{corollary}[theorem]{Corollary}
\newtheorem{lemma}[theorem]{Lemma}

\newtheorem{definition}[theorem]{Definition}

\newenvironment{proofof}[1]{\begin{trivlist} \item {\bf Proof
#1:~~}}
  {\qed\end{trivlist}}

\newcommand{\namedref}[2]{\hyperref[#2]{#1~\ref*{#2}}}

\newcommand\norm[1]{\left\lVert#1\right\rVert}
\newcommand{\PPr}[1]{\ensuremath{\mathbf{Pr}\left[#1\right]}}

\newcommand{\Ex}[1]{\ensuremath{\mathbb{E}\left[#1\right]}}

\renewcommand{\O}[1]{\ensuremath{\mathcal{O}\left(#1\right)}}
\newcommand{\tO}[1]{\ensuremath{\tilde{\mathcal{O}}\left(#1\right)}}
\newcommand{\eps}{\varepsilon}

\newcommand{\lprp}[1]{\left(#1\right)}
\newcommand{\lbrb}[1]{\left\{#1\right\}}
\newcommand{\lsrs}[1]{\left[#1\right]}

\newcommand{\mb}[1]{\mathbb{#1}}
\newcommand{\mc}[1]{\mathcal{#1}}
\newcommand{\mrm}[1]{\mathrm{#1}}

\DeclarePairedDelimiterX{\inp}[2]{\langle}{\rangle}{#1, #2}

\newcommand{\N}{\mathbb{N}}
\newcommand{\R}{\mathbb{R}}
\newcommand{\Exp}{\mathbb{E}}
\renewcommand{\P}{\mathbb{P}}

\def \ba    {\mdef{\mathbf{a}}}

\def \calE    {\mdef{\mathcal{E}}}

\def \calS    {\mdef{\mathcal{S}}}

\def \bA    {\mdef{\mathbf{A}}}
\def \bB    {\mdef{\mathbf{B}}}
\def \boC    {\mdef{\mathbf{C}}}
\def \bD    {\mdef{\mathbf{D}}}
\def \bG    {\mdef{\mathbf{G}}}
\def \bH    {\mdef{\mathbf{H}}}
\def \bI    {\mdef{\mathbb{I}}}
\def \bM    {\mdef{\mathbf{M}}}

\def \bR    {\mdef{\mathbf{R}}}
\def \bS    {\mdef{\mathbf{S}}}
\def \bSigma    {\mdef{\mathbf{\Sigma}}}

\def \bQ    {\mdef{\mathbf{Q}}}
\def \bV    {\mdef{\mathbf{V}}}
\def \bU    {\mdef{\mathbf{U}}}

\def \bY    {\mdef{\mathbf{Y}}}
\def \bZ    {\mdef{\mathbf{Z}}}
\def \bW    {\mdef{\mathbf{W}}}
\def \bb    {\mdef{\mathbf{b}}}
\def \be    {\mdef{\mathbf{e}}}
\def \bbm    {\mdef{\mathbf{m}}}

\def \bu    {\mdef{\mathbf{u}}}
\def \bw    {\mdef{\mathbf{w}}}
\def \bv    {\mdef{\mathbf{v}}}
\def \bx    {\mdef{\mathbf{x}}}
\def \by    {\mdef{\mathbf{y}}}
\def \bz    {\mdef{\mathbf{z}}}
\def \bg    {\mdef{\mathbf{g}}}

\newcommand{\mdef}[1]{{\ensuremath{#1}}\xspace}  
\DeclareMathOperator*{\argmin}{argmin}

\DeclareMathOperator*{\polylog}{polylog}

\DeclareMathOperator*{\poly}{poly}

\DeclareMathOperator*{\nnz}{nnz}
\DeclareMathOperator*{\rank}{rank}

\newcommand{\abs}[1]{\mdef{\left|#1\right|}}         
\newcommand{\E}[2][]{\mdef{\underset{#1}{\mathbb{E}}\left[#2\right]}} 

\newcommand{\ignore}[1]{}

\newif\ifnotes\notestrue 
\ifnotes
\newcommand{\samson}[1]{\textcolor{purple}{{\bf (Samson:} {#1}{\bf ) }} \marginpar{\tiny\bf
             \begin{minipage}[t]{0.5in}
               \raggedright S:
            \end{minipage}}}            							
\else
\newcommand{\samson}[1]{}
\fi

\ifnotes
\newcommand{\sandeep}[1]{\textcolor{red}{{\bf (Sandeep:} {#1}{\bf ) }} \marginpar{\tiny\bf
             \begin{minipage}[t]{0.5in}
               \raggedright S:
            \end{minipage}}}            							
\else
\newcommand{\sandeep}[1]{}
\fi

\hypersetup{
     colorlinks   = true,
     citecolor    = blue,
		 linkcolor		= red
}

\makeatletter
\renewcommand*{\@fnsymbol}[1]{\textcolor{mahogany}{\ensuremath{\ifcase#1\or *\or \dagger\or \ddagger\or
 \mathsection\or \triangledown\or \mathparagraph\or \|\or **\or \dagger\dagger
   \or \ddagger\ddagger \else\@ctrerr\fi}}}
\makeatother

\providecommand{\email}[1]{\href{mailto:#1}{\nolinkurl{#1}\xspace}}

\definecolor{mahogany}{rgb}{0.75, 0.25, 0.0}
\definecolor{darkblue}{rgb}{0.0, 0.0, 0.55}
\definecolor{darkpastelgreen}{rgb}{0.01, 0.75, 0.24}
\definecolor{bleudefrance}{rgb}{0.19, 0.55, 0.91}
\hypersetup{
     colorlinks   = true,
     citecolor    = darkblue,
  	 linkcolor	  = mahogany
}

\title{Optimal Algorithms for Linear Algebra in the Current Matrix Multiplication Time}

\author{
Yeshwanth Cherapanamjeri\thanks{UC Berkeley. 
E-mail: \email{yeshwanth@berkeley.edu}}
\and
Sandeep Silwal\thanks{MIT. 
E-mail: \email{silwal@mit.edu}}
\and
David P. Woodruff\thanks{Carnegie Mellon University. 
E-mail: \email{dwoodruf@cs.cmu.edu}}
\and
Samson Zhou\thanks{UC Berkeley and Rice University. 
Work done in part while at Carnegie Mellon University. 
E-mail: \url{samsonzhou@gmail.com}}
}
\date{\today}
\date{}

\begin{document}
\allowdisplaybreaks
\maketitle

\begin{abstract}
We study fundamental problems in linear algebra, such as finding a maximal linearly independent subset of rows or columns (a basis), solving linear regression, or computing a subspace embedding. For these problems, we consider input matrices  $\mathbf{A}\in\mathbb{R}^{n\times d}$ with $n > d$. The input can be read in $\text{nnz}(\mathbf{A})$ time, which denotes the number of nonzero entries of $\mathbf{A}$. In this paper, we show that beyond the time required to read the input matrix, these fundamental linear algebra problems can be solved in $d^{\omega}$ time, i.e., where $\omega \approx 2.37$ is the current matrix-multiplication exponent.

To do so, we introduce a constant-factor subspace embedding with the optimal $m=\mathcal{O}(d)$ number of rows, and which can be applied in time $\mathcal{O}\left(\frac{\text{nnz}(\mathbf{A})}{\alpha}\right) + d^{2 + \alpha}\text{poly}(\log d)$ for any trade-off parameter $\alpha>0$, tightening a recent result by Chepurko et. al. [SODA 2022] that achieves an $\exp(\text{poly}(\log\log n))$ distortion with $m=d\cdot\text{poly}(\log\log d)$ rows in $\mathcal{O}\left(\frac{\text{nnz}(\mathbf{A})}{\alpha}+d^{2+\alpha+o(1)}\right)$ time. Our subspace embedding uses a recently shown property of {\it stacked} Subsampled Randomized Hadamard Transforms (SRHT), which actually increase the input dimension, to ``spread'' the mass of an input vector among a large number of coordinates, followed by random sampling. To control the effects of random sampling, we use fast semidefinite programming to reweight the rows. We then use our constant-factor subspace embedding to give the first optimal runtime algorithms for finding a maximal linearly independent subset of columns, regression, and leverage score sampling. To do so, we also introduce a novel subroutine that iteratively grows a set of independent rows, which may be of independent interest.
\end{abstract}

\section{Introduction}
In this paper, we consider fundamental problems in linear algebra, such as finding a maximal linearly independent subset of rows or columns, i.e., a basis, or solving linear regression, or computing a subspace embedding. Surprisingly, we still do not have optimal algorithms for these tasks. 
The input to these problems is generally a matrix $\bA\in\mathbb{R}^{n\times d}$ with $n > d$ that requires $\nnz(\bA)$ time to read, where $\nnz(\bA)$ is the number of non-zero entries of $\bA$. 
Algorithms for these problems frequently use subroutines such as matrix multiplication, inverse computation, or decomposition (singular value, QR, LU, etc.), that use at least $nd^{\omega-1}$ time, where $\omega\approx2.37$ is the exponent for matrix multiplication~\cite{AlmanW21}.

Dimensionality reduction techniques are often utilized to decrease the effective input size, so that a solution to the smaller input is often a good approximation to the optimal solution of the original problem. 
These approaches transform $\bA$ into a matrix $\bM\in\mathbb{R}^{m\times d}$ with $m\ll n$ and (approximately) solve the problem on the instance $\bM$ with a significantly smaller number of rows. 
However, existing results could only achieve $m=d\polylog(d)$ for dimensionality reduction techniques with input-sparsity runtime, which prevented true matrix-multiplication runtime algorithms, i.e., running times of the form $\O{d^\omega}$. 

Here we emphasize that $\omega$ is the parameter between $2$ and $3$ for the matrix multiplication exponent, possibly depending on the input parameters for matrix multiplication, e.g., matrix multiplication between two $n\times n$ matrices uses $\O{n^\omega}$ time, for some fixed matrix multiplication oracle that we are given. 
By contrast, many previous works define $\omega$ to be the smallest \emph{constant} such that matrix multiplication between two $n\times n$ matrices runs in time $\O{n^{\omega+\eps}}$ for any constant $\eps>0$. 
In particular, \cite{CoppersmithW82} showed that given a matrix multiplication algorithm with runtime  $\O{n^{\omega+\eps_1}}$, there exists a matrix multiplication algorithm with runtime $\O{n^{\omega+\eps_2}}$ with $\eps_2\in(0,\eps_1)$ and this process can continue ad infinitum. 
However, at some point we require an explicit fixed matrix multiplication algorithm for downstream applications. 
Thus, we consider access to a fixed matrix multiplication algorithm with runtime $\O{n^\omega}$, so that it is important to track and eliminate additional $\polylog$ overheads on top of the matrix multiplication \emph{runtime of the fixed algorithm}.  
Indeed, the removal of the last logarithmic factors is related to well-known conjectures on the construction of sparse Johnson-Lindenstrauss transforms~\cite{NelsonN13,ChepurkoCKW22}. 
In a recent work, \cite{ChepurkoCKW22} showed that these logarithmic factors were not necessary, achieving algorithms for linear algebra in near matrix-multiplication runtime, up to $\poly(\log\log d)$ factors. 

\subsection{Our Contribution}
In this work, we give the first algorithms for linear algebra in true matrix-multiplication runtime, removing the last $\poly(\log\log d)$ factors in the algorithms of \cite{ChepurkoCKW22} and thus closing a long line of work. 
Our results show that beyond the time required to read the input matrix, fundamental linear algebra problems such as finding a maximal linearly independent subset of rows or columns (a basis), linear regression, or computing a subspace embedding can be solved in the current matrix-multiplication runtime. 

We first introduce a constant-factor subspace embedding that uses input-sparsity runtime:
\begin{restatable}{theorem}{thmse}
\label{thm:se}
For any $\bA \in \R^{n \times d}$ and any tradeoff parameter $\alpha > 0$, we can compute matrix $\bG \in \R^{p \times n}$ such that:
\[\forall \bx \in \R^d: \norm{\bA\bx}_2 \leq \norm{\bG\bA \bx}_2 \leq \xi \norm{\bA\bx}_2,\]
with probability at least $0.9$ for a fixed constant $\xi > 1$. Furthermore, we have $p = \O{d}$ and $\bG\bA$ may be computed in time $\O{\frac{\nnz(\bA)}{\alpha}} + d^{2 + \alpha} \polylog (d)$.
\end{restatable}
By comparison, \cite{ChepurkoCKW22} recently gave a subspace embedding with distortion $\exp(\poly(\log\log n))$ using runtime $\O{\frac{\nnz(\bA)}{\alpha}+d^{2+\alpha+o(1)}}$, for any tradeoff parameter $\alpha>0$. 
Theorem \ref{thm:se} also extends naturally to the case when $\bA$ has rank $k$, in which case it suffices for $\bG \in \R^{p \times n}$ to have $p=\O{k}$ rows and the resulting time to compute $\bG\bA$ is $\O{\frac{\nnz(\bA)}{\alpha}} + k^{2 + \alpha} \polylog(k)$.

Our constant-factor subspace embedding in Theorem \ref{thm:se} uses a sampling procedure that leverages a recently observed property of stacked Subsampled Randomized Hadamard Transforms (SRHTs) to ``spread'' the mass of an input vector among a large number of coordinates~\cite{CherapanamjeriN22}. 
Our constant-factor subspace embedding can be used to improve the efficiency of leverage score sampling, which has applications in a number of important linear algebra problems~\cite{Woodruff14}. 

In particular, we can further boost our constant-factor subspace embedding to a $(1+\eps)$-approximate subspace embedding through leverage score sampling:
\begin{restatable}{theorem}{thmalgepsse}
\label{thm:alg:eps:se}
Given $\bA\in\mathbb{R}^{n\times d}$, an accuracy parameter $\eps>0$, and any tradeoff parameter $\alpha>0$, there exists an algorithm that computes a matrix $\bS\bA$ with $\O{\frac{1}{\eps^2}\,d\log d}$ rows such that with probability at least $\frac{9}{10}$, for all vectors $\bx\in\mathbb{R}^d$,
\[(1-\eps)\|\bA\bx\|_2\le\|\bS\bA\bx\|_2\le(1+\eps)\|\bA\bx\|_2.\]
Moreover, $\bS\bA$ can be computed in time
\[\O{\frac{\nnz(\bA)}{\alpha}+d^\omega}+\frac{1}{\eps^2}\,d^{2+\alpha}\polylog(d).\]
\end{restatable}
By comparison, recent work by \cite{ChepurkoCKW22} achieved a $(1+\eps)$-subspace embedding with either $\frac{1}{\eps^2}\,(d\log d)\exp(\poly(\log\log d))$ rows or with runtime $\O{\frac{\nnz(\bA)}{\alpha}}+d^\omega\poly(\log\log d)+\frac{1}{\eps^3}\,d^{2+o(1)}+\frac{1}{\eps^2}\,n^{\alpha+o(1)}d^{2+o(1)}$. 
Our result avoids such tradeoffs, which is especially useful in downstream applications, as we soon discuss. 
Moreover, Theorem \ref{thm:alg:eps:se} extends naturally to the case where $\bA$ has rank $k$, similarly as Theorem \ref{thm:se}. 
On the other hand, we remark that unlike our constant-factor subspace embedding, our $(1+\eps)$-subspace embedding \emph{does not} have the optimal number of rows to perform further tasks downstream. 
We believe the existence/construction of such a subspace embedding would be an interesting future question. 
Questions in a similar spirit have also been previously asked for graph theoretic problems, e.g.,~\cite{LeeS17,LeeS18}. 

We then use our constant-factor subspace embedding and our leverage score sampling framework to find a maximal set of linearly independent rows of an input matrix $\bA\in\mathbb{R}^{n\times d}$:
\begin{restatable}{theorem}{thmindrows}
\label{thm:ind:rows}
Given a matrix $\bA\in\mathbb{R}^{n\times d}$ with rank $k$ and any tradeoff parameter $\alpha>0$, there exists an algorithm that outputs a set of $k$ linearly independent rows of $\bA$, using time $\O{\frac{\nnz(\bA)}{\alpha}+k^\omega}+k^{2+\alpha}\polylog(k)$
\end{restatable}
By comparison, recent work of \cite{ChepurkoCKW22} gave an algorithm that finds a set of $k$ linearly independent rows of $\bA$ using time $\O{\nnz(\bA)+k^{2+o(1)}}+k^\omega\,\poly(\log\log k)$. 
Like the algorithm of \cite{ChepurkoCKW22}, we first reduce the problem to computing a set of $k$ linearly dependent rows of a matrix $\bB\in\mathbb{R}^{\O{k\log k}\times\O{k}}$, though due to our more efficient subspace embedding algorithm, we can do this in matrix-multiplication runtime while \cite{ChepurkoCKW22} cannot. 
Now to achieve Theorem \ref{thm:ind:rows}, we develop a novel subroutine in Section \ref{sec:basis:grow}, which iteratively grows a set of independent rows of $\bB$ that may be of independent interest. 
Crucially, the algorithm avoids additional $k^\omega\,\poly(\log\log k)$ dependencies that are incurred by the subroutine of \cite{ChepurkoCKW22}. 
We provide a summary of previous work on finding a set of independent rows in Figure \ref{fig:ind:rows:}. 

\begin{figure*}[!htb]
\centering
\begin{tabular}{|c|c|}\hline
Reference & Runtime \\\hline
Gaussian elimination & $\O{nd^{\omega-1}}$ \\\hline
\cite{CheungKL13} & $\O{\nnz(\bA)\log n+k^{\omega}\log n}$ \\\hline
\cite{ChepurkoCKW22} & $\O{\nnz(\bA)+k^{2+o(1)}}+k^\omega\,\poly(\log\log k)$ \\\hline
Theorem \ref{thm:ind:rows}, for any $\alpha>0$ & $\O{\frac{\nnz(\bA)}{\alpha}+k^\omega}+k^{2+\alpha}\polylog(k)$ \\\hline
\end{tabular}
\caption{Summary of previous results for identifying a set of $k$ independent rows from a matrix $\bA\in\mathbb{R}^{n\times d}$ with rank $k$}
\label{fig:ind:rows:}
\end{figure*}

Finally, we use our $(1+\eps)$-approximate subspace embedding to achieve $(1+\eps)$-approximate linear regression:
\begin{restatable}{theorem}{thmlinreg}
\label{thm:lin:reg}
Given $\bA\in\mathbb{R}^{n\times d}$, $\bb\in\mathbb{R}^n$, and any tradeoff parameter $\alpha>0$, there exists an algorithm that with probability at least $0.9$, outputs a vector $\by$ such that
\[\|\bA\by-\bb\|_2\le(1+\eps)\min_{\bx\in\mathbb{R}^d}\|\bA\bx-\bb\|_2,\]
using time $\O{\frac{\nnz(\bA)}{\alpha}+d^\omega}+\frac{1}{\eps}\,d^{2+\alpha}\polylog(d)+\frac{1}{\eps}\,d^2\polylog(d)\log\frac{1}{\eps}$. 
\end{restatable}
By comparison, \cite{ChepurkoCKW22} output a $(1+\eps)$-approximation to linear regression in time $\O{\frac{\nnz(\bA)}{\alpha}}+\frac{1}{\eps^3}\,n^{\alpha+o(1)}d^{2+o(1)}+d^\omega\,\poly(\log\log d)$ for any tradeoff parameter $\alpha>0$. 
We remark that our result achieves both the optimal (current) matrix-multiplication runtime and also a better dependence on $\frac{1}{\eps}$, which is a byproduct of our $(1+\eps)$-subspace embedding avoiding the tradeoffs incurred by \cite{ChepurkoCKW22}. 

\subsection{Technical Overview}
In this section, we give a brief overview of our technical contributions. 
A summary of the interplay between our algorithmic contributions can be seen in Figure \ref{fig:flowchart}. 

To avoid polylogarithmic overhead over matrix-multiplication runtime, we require a dimensionality reduction technique that uses $o(d\log d)$ rows. 
A folklore result states that a dense matrix $\bG$ of $\O{d}$ rows with entries that are independent Sub-Gaussian random variables suffices to achieve a constant-factor subspace embedding~\cite{Woodruff14}. 
However, computing $\bG\bA$ for an input matrix $\bA\in\mathbb{R}^{n\times d}$ requires time $\O{d\,\nnz(\bA)}$ due to the multiplication with the dense matrix $\bG$. 
Although the runtime is unsatisfactory, a subspace embedding with a small number of rows corresponds to an improvement in runtime for downstream tasks. 
Unfortunately, faster dimensionality reduction techniques such as the sparse Johnson-Lindenstrauss transform use $\Omega(d\log d)$ rows and it is an open question whether these constructions can be improved to only using $\O{d}$ rows~\cite{NelsonN13}.

\paragraph{Intuition from previous work.}
\cite{ChepurkoCKW22} recently sidestepped these issues by first showing that a rescaling of an embedding of \cite{Indyk07} roughly maintains the $\ell_2$ norm of a $d$-dimensional unit $\ell_2$ vector while the $\ell_1$ norm becomes $\tilde{\Omega}(\sqrt{d})$. 
In particular, a constant fraction of the resulting $o(d\log d)$ coordinates have magnitude $\tilde{\Omega}\left(\frac{1}{\sqrt{d}}\right)$. 
\cite{ChepurkoCKW22} then used the intuition that under such a ``flattening'' $\by$ of a unit $\ell_2$ vector, a sparse matrix $\bS$ of random signs will sample some of the coordinates with ``large'' magnitude, so that the dot product $\langle\bS_i,\by\rangle$ of each row of $\bS$ with the flattened vector $\by$ will be at least $\tilde{\Omega}\left(\frac{1}{\sqrt{d}}\right)$, which implies a lower bound on $\|\bS\by\|_2^2$. 
An upper bound on the operator norm $\|\bS\|_2$ can also be shown, ultimately giving a distortion of $\exp(\poly(\log\log d))$. 
Unfortunately, Indyk's embedding~\cite{Indyk07} only governs the sum of the $\ell_2$ norms of blocks of coordinates of the resulting embedding and thus must be applied recursively across $\O{\log\log d}$ levels, resulting in a sketching matrix with $d\,\exp(\log\log(d))$ rows. 
Therefore, the resulting matrix does not lose polylogarithmic factors over matrix-multiplication runtime, but still cannot quite achieve true matrix-multiplication runtime. 
Hence, using Indyk's embedding~\cite{Indyk07} seems to be a major bottleneck to achieving matrix multiplication runtime in the previous works and thus it seems we need to use significantly new techniques altogether. 

\paragraph{Crude subspace embeddings through SRHT.} 
To avoid the extraneous factors over matrix-multiplication runtime, we require a sketching matrix with a smaller number of rows which can also be applied quickly. 
To that end, we recall that the flattening property of the Subsampled Randomized Hadamard Transform (SRHT) is frequently used to show that it forms a subspace embedding, in the sense that the largest coordinate of the image of a fixed vector is upper bounded by $\tO{\frac{1}{\sqrt{d}}}$. 
More recently, \cite{CherapanamjeriN22} showed that by stacking some number of SRHTs on top of each other to {\it increase} dimension, not only is the $\ell_2$-norm of any fixed vector preserved exactly, but also for any vector, a constant fraction of the coordinates of the image of the SRHT is lower bounded by $\tilde{\Omega}\left(\frac{1}{\sqrt{d}}\right)$ in absolute value.  
This property can be seen as a fast embedding of $\ell_2$ into $\ell_1$ with a small target dimension. 

We first left-multiply our input matrix $\bA\in\mathbb{R}^{n\times d}$ with OSNAP matrices \cite{NelsonN13} $\bS_1$ and $\bS_2$ to obtain a constant-factor subspace embedding.
The application of the composition of OSNAP matrices is a standard technique that allows us to first decrease the number of rows to $\ell:=\O{d\log d}$ at the cost of a constant-factor distortion, in $\O{\nnz(\bA) + d^C}$ time, where $C$ can be made an arbitrarily small constant larger than $2$. 

If $\bM\in\mathbb{R}^{m\ell\times\ell}$ is the matrix consisting of stacked randomized Hadamard matrices, we define $\bS\in\mathbb{R}^{\O{d}\times m\ell}$ to be a matrix that independently and uniformly samples each row of $\bM\bS_1\bS_2\bA$. 
It then suffices to bound both the contraction and the dilation of $\bS\bM\bB:=\bS\bM\bS_1\bS_2\bA$. 
That is, if we could show $\|\bS\bM\|_2\le\O{1}$ and $\|\bS\bM\bB\bx\|_2\ge\O{1}$ for all unit vectors $\bB\bx\in\mathbb{R}^\ell$, then it follows that $\bS\bM$ is a constant factor subspace embedding. 
To handle contraction, we show that $\|\bS\bM\bB\bx\|^2_2\ge\O{1}$ by first showing concentration for a single vector $\bx\in\mathbb{R}^\ell$ due to the abundance of ``large'' coordinates in $\bM\bB\bx$, followed by taking a union bound over a sufficiently fine net. 
Unfortunately, it does not seem evident how to bound the dilation by a constant (or whether it is even true). 
For instance, using either a crude concentration inequality or more sophisticated results bounding the norms of random submatrices, e.g., \cite{tropp2008norms}, seems to give an extra logarithmic factor. 
Thus this approach yields a subspace embedding with a logarithmic distortion, which is not enough for our optimal algorithms for downstream tasks.  

\paragraph{Constant-factor subspace embedding.} 
To achieve a constant-factor subspace embedding, we note that the slack in the above analysis is that there are large entries in $\bM$ that can be sampled by $\bS$, which prevents constant upper bounds on $\|\bS\bM\|_2$. 
On the other hand, due to the abundance of large coordinates in $\bM\bB\bx$, an accurate estimation can still be acquired without these large entries in $\bM$. 
Indeed, we show that with high probability, there exists a slight reweighting $\bW$ of the sampled rows (possibly with weight $0$ to remove the large entries) such that the resulting reweighted subsampled matrix has operator bounded by a constant, i.e., $\|\bW\bS\bM\|_2\le\O{1}$. 

Moreover, we show that the reweighting can be efficiently computed by solving a standard packing semidefinite program (SDP). 
Crucially, the SDP requires a fast projection oracle, which we can implement due to our subspace embedding with a logarithmic distortion discussed above. 
Finally, we show that under this reweighting, the contraction does not drastically change, so that $\|\bW\bS\bM\bB\bx\|_2\ge\O{1}$ for all unit vectors $\bB\bx\in\mathbb{R}^n$. 
Hence, it follows that $\bW\bS\bM\bB$ is a constant-factor subspace embedding. 

We remark that we require sharper bounds in downstream applications, e.g., basis selection, when $\bA\in\mathbb{R}^{n\times d}$ is not full rank. 
In this case, our algorithms naturally generalize to dimension and runtime dependent on the rank of $\bA$ rather than the dimension $d$ of $\bA$. 
We summarize our constant-factor subspace embedding at a high level in Figure \ref{fig:constant:se}. 

\begin{figure*}
\centering
\begin{mdframed}
Given input matrix $\bA\in\mathbb{R}^{n\times d}$:
\begin{enumerate}
\item 
Apply OSNAP matrix $\bS_2\in\mathbb{R}^{n'\times n}$ with constant factor distortion and tradeoff-parameter $\alpha$ to acquire $\bS_2\bA\in\mathbb{R}^{n'\times d}$, where $n'=\O{d^{1+\alpha}\log d}$
\item
Apply OSNAP matrix $\bS_1\in\mathbb{R}^{\ell\times n'}$ with constant factor distortion and tradeoff-parameter $\alpha'=\frac{1}{\log d}$ to acquire $\bS_1\bS_2\bA\in\mathbb{R}^{\ell\times d}$, where $\ell=\O{d\log d}$
\item
Apply SRHT matrix $\bM\in\mathbb{R}^{m\ell\times\ell}$ for $m=\polylog(d)$ to acquire $\bM\bS_1\bS_2\bA\in\mathbb{R}^{m\ell\times d}$
\item
Apply a sampling matrix $\bS\in\mathbb{R}^{p\times m\ell}$ that uniformly samples rows with $p=\O{d}$, to acquire $\bS\bM\bS_1\bS_2\bA\in\mathbb{R}^{d'\times d}$, where $d'=\O{d}$ with high probability
\item
Solve an SDP to find a set of weights $\bW\in\mathbb{R}^{d'\times d'}$ so that the operator norm of $\bS\bM$ is appropriately bounded and output $\bW\bS\bM\bS_1\bS_2\bA\in\mathbb{R}^{d'\times d}$
\end{enumerate}
\end{mdframed}
\caption{High-level summary of our constant-factor subspace embedding}
\label{fig:constant:se}
\end{figure*}

\paragraph{Leverage score sampling.}
To achieve a $(1+\eps)$-subspace embedding for a matrix $\bA$, a standard approach is to perform leverage score sampling, i.e., to sample $\O{\frac{1}{\eps^2}\,d\log d}$ rows of $\bA$ with probabilities proportional to their leverage score.  
However existing techniques could not be run in matrix-multiplication runtime, and instead ran in at least $\O{d^{\omega}\log d}$ time. 

We instead use our fast constant-factor subspace embedding $\bS\bA$ into an optimal target dimension. 
We can also efficiently compute its QR decomposition so that $\bS\bA=\bQ\bR^{-1}$ for a matrix $\bQ$ with orthonormal columns. 
In particular, since $\bQ$ has orthonormal columns, then the leverage scores of $\bS\bA$ are precisely the squared row norms of $\bQ$. 
It follows that the squared row norm of $\ba_i\bR$ is a constant-factor approximation to the leverage score of row $\ba_i$ for each $i\in[n]$. 
We can then apply the standard leverage score sampling approach by sampling $\O{\frac{1}{\eps^2}\,d\log d}$ rows of $\bA$ with probabilities proportional to their leverage score to achieve a $(1+\eps)$-subspace embedding in matrix-multiplication runtime. 
In particular, we first compute $\bA\bR\bG$ for a Johnson-Lindenstrauss matrix $\bG$ and then perform rejection sampling to achieve leverage score sampling (see Theorem~\ref{thm:produce:sample} and surrounding discussion for more details). 
We also remark that we only require the orthogonality of $\bQ$ in the QR decomposition, so other methods such as the SVD decomposition that yield a matrix with orthonormal columns would also suffice. 

We again emphasize that unlike our constant-factor subspace embedding, our $(1+\eps)$-subspace embedding \emph{does not} have the optimal number of rows to perform further tasks downstream. 
That is, our $(1+\eps)$-subspace embedding uses $\O{\frac{1}{\eps^2}\,d\log d}$ rows. 
By comparison, our constant-factor subspace embedding uses $\O{d}$ rows, which is better for $\eps=\O{1}$. 
Thus our result should be interpreted as the ability to perform leverage score sampling in matrix-multiplication runtime, with one such application being a $(1+\eps)$-subspace embedding and another such application being the selection of an independent basis (see below). 
An interesting open question is whether our techniques can be further refined to achieve a $(1+\eps)$-subspace embedding with $\O{\frac{d}{\eps^2}}$ rows in matrix-multiplication runtime. 

\paragraph{Basis selection.}
To find a set of $k$ independent rows for an input matrix $\bA\in\mathbb{R}^{n\times d}$ with rank $k$, we first use our efficient leverage score sampling framework in conjunction with existing techniques to reduce the effective input to size $\mathbb{R}^{\O{k\log(k)}\times\O{k}}$. 
Namely, we first note that there exists a distribution of matrices that form a rank-preserving sketch, so that $\rank(\bA)=\rank(\bA\bS)$, where $\bS\in\mathbb{R}^{d\times ck}$ for some constant $c>0$. 
Moreover the rank-preserving sketch has the property that any set of independent rows of $\bA$ is also a set of independent rows of $\bA\bS$ and vice versa. 
Thus it would suffice to select a basis of rows from $\bA\bS$ and take the corresponding rows in $\bA$. 

However, we cannot explicitly compute $\bA\bS$. 
We also cannot use our constant-factor subspace embedding, because multiplication by a Hadamard matrix distorts the mapping of the indices of independent rows in the original matrix and the indices of independent rows in the smaller matrix. 
Instead, we use our constant-factor subspace embedding, which selects $\O{k\log k}$ reweighted rows from $\bA\bS$. 
Hence, it remains to select a basis of rows from a matrix $\bB\in\mathbb{R}^{\O{k\log(k)}\times\O{k}}$ with rank $k$.

For this sub-problem, there exist a number of previous techniques, such as an approach by \cite{ChepurkoCKW22} that iteratively removes redundant rows from $\bB$. 
We remark that these techniques are generally not optimized to run in time $\O{k^\omega}$ since other components are usually a larger bottleneck, and thus it seems we require a new set of techniques. 

\paragraph{Iteratively growing a basis.}
To select a basis of rows from a matrix $\bB\in\mathbb{R}^{\O{k\log(k)}\times\O{k}}$ with rank $k$, we develop a new algorithm that iteratively grows a set $S$ of independent rows of $\bB$. 
Namely, we use leverage score sampling to sample $\O{k}$ rows of $\bB$. 
Observe that this is not enough to cover the entire row span of $\bB$, but for $c=\frac{1}{10}$, we show using approximate matrix product on a well-conditioned verion of our input, that we can get rank $(1-c)k=\frac{9}{10}k$ with probability at least $\frac{2}{3}$. 
We can add an independent subset of these rows to our growing set $S$ and then compute a basis $\bZ_1$ for the orthogonal complement of $S$. 
So far, these procedures, i.e., leverage score sampling, independent subset selection, and orthogonal complement basis computation, all use at most $\gamma k^\omega$ runtime for an absolute constant $\gamma>0$.  

We now repeat these procedures on $\bB\bZ_1^\top$, first using leverage score sampling to sample $\O{k}$ rows of $\bB\bZ_1^\top$. 
An observation is that the rows of $\bB$ that are spanned by $S$ will all be zero in $\bB\bZ_1^\top$, since $\bZ_1$ is the orthogonal complement of $S$. 
Thus again with probability $\frac{2}{3}$, we can sample a set of rows with rank at least a $\frac{9}{10}$ fraction of the rank of $\bB\bZ_1^\top$. 
We can again add an independent subset of these rows to our growing set $S$ and then compute a basis $\bZ_2$ for the orthogonal complement of $S$. 
However since conditioned on the success of the previous iteration, the rank of $\bB\bZ_1^\top$ is at most a $c=\frac{1}{10}$-fraction of the rank of $\bB$, these procedures will now take at most $\gamma(ck)^\omega$ runtime, which is a constant fraction smaller.

We can thus proceed by iteratively adding rows of $\bB$ to $S$ until $S$ has rank $k$. 
We can then output the corresponding rows of $\bA$. 
The runtime in each iteration, conditioned on successful samplings in each previous iteration, follows a geometric series and thus the overall runtime is $\O{k^\omega}$. 
The runtime analysis is also robust to failures in each iteration because a failure in an iteration means that at worst, no additional rows are added to $S$. 
Therefore, the algorithm will always terminate with a set of independent rows and we can simply compute the expected runtime of this procedure, which still follows a geometric series. 

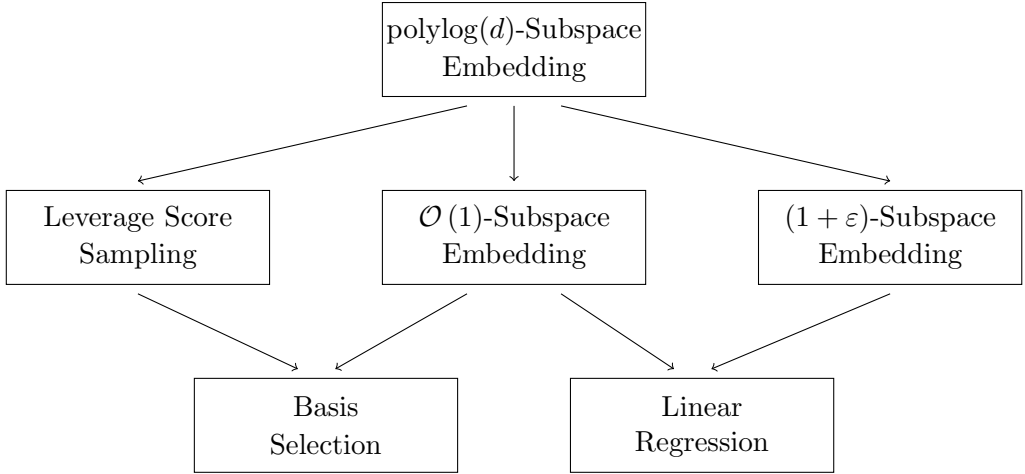
\begin{figure*}[!htb]
\centering
\begin{tikzpicture}[scale=1.25]

\node at (4,4){$\polylog(d)$-Subspace};
\node at (4,3.6){Embedding};
\draw (2.6,3.3) rectangle+(2.8,1); 

\draw[->] (4,3.2) -- (4,2.4);
\node at (0,2){Leverage Score};
\node at (0,1.6){Sampling};
\draw (-1.4,1.3) rectangle+(2.8,1); 

\draw[->] (3.5,3.2) -- (0,2.4);
\node at (4,2){$\O{1}$-Subspace};
\node at (4,1.6){Embedding};
\draw (2.6,1.3) rectangle+(2.8,1); 

\draw[->] (4.5,3.2) -- (8,2.4);
\node at (8,2){$(1+\eps)$-Subspace};
\node at (8,1.6){Embedding};
\draw (6.6,1.3) rectangle+(2.8,1); 

\draw[->] (3.5,1.2) -- (2.1,0.4);
\draw[->] (0,1.2) -- (1.7,0.4);
\node at (2,0){Basis};
\node at (2,-0.4){Selection};
\draw (0.6,-0.7) rectangle+(2.8,1); 

\draw[->] (4.5,1.2) -- (5.7,0.4);
\draw[->] (8,1.2) -- (6.1,0.4);
\node at (6,0){Linear};
\node at (6,-0.4){Regression};
\draw (4.6,-0.7) rectangle+(2.8,1); 

\end{tikzpicture}
\caption{Flowchart of dependencies for our algorithmic contributions.}
\label{fig:flowchart}
\end{figure*}

\paragraph{Linear regression.}
For linear regression, we would like to find a vector $\by$ such that
\[\|\bA\by-\bb\|_2\le(1+\eps)\min_{\bx\in\mathbb{R}^d}\|\bA\bx-\bb\|_2.\]
We can first compute a $(1+\eps)$-approximate subspace embedding\footnote{It is known that even a $(1+\O{\sqrt{\eps}})$-approximate subspace embedding suffices, see Lemma \ref{lem:se:lr:solve}. To facilitate the intuition for our algorithm, we defer this discussion to Section \ref{sec:lin:reg}.} $[\bS\bA;\bS\bb]$ of the matrix $[\bA;\bb]$. 
However, the subspace embedding $[\bS\bA;\bS\bb]$ has $\frac{1}{\eps^2}\,d\polylog(d)$ rows and so we cannot directly solve for the optimal solution on the smaller space, since computing the closed-form solution would not be true matrix-multiplication runtime. 
On the other hand, we only require finding an approximately optimal solution on the smaller space, i.e., we only want a vector $\bw\in\mathbb{R}^d$ such that $\|\bS\bA\bw-\bS\bb\|_2\le(1+\O{\eps})\min_{\bx\in\mathbb{R}^d}\|\bS\bA\bx-\bS\bb\|_2$. 
Thus we instead use gradient descent to find such a vector $\bw$. 

For efficient runtime, gradient descent requires a small condition number and a ``good'' initial solution. 
While $\bS\bA$ is a $(1+\eps)$-approximate subspace embedding, it may not necessarily have small condition number. 
To decrease the condition number to $\O{1}$, we instead consider $\min_{\bx\in\mathbb{R}^d}\|\bS\bA\bR\bx-\bS\bb\|_2$, where $\bG\bA=\bQ\bR^{-1}$ is a QR decomposition for a constant-factor subspace embedding $\bG\bA$. 
Thus in this setting, $\bR$ can be considered as a preconditioner. 
To find a good initial solution, we first find the closed-form solution to $\bw^{(0)}=\argmin_{\bx\in\mathbb{R}^d}\|\bG\bA\bx-\bG\bb\|_2$, since $\bG\bA$ is a constant-factor subspace embedding. 
It then remains to account for the preconditioning by computing $\bw^{(1)}=\bR^{-1}\bw^{(0)}$, which is a good initial solution for gradient descent, since it provides a constant-factor approximation to the optimal solution due to the properties of $\bG\bA$. 

Instead, we use the SRHT to compute a matrix $\bG$ that is a constant-factor subspace embedding, with $\O{d}$ rows. 
By computing $\bG\bA=\bQ\bR^{-1}$, we can compute a matrix $\bR$ such that $\kappa(\bA\bR)=\O{1}$, since $\bG$ is a constant-factor subspace embedding and $\kappa(\bG\bA\bR)=1$ due to the orthogonality of $\bQ$. 
This implies that an approximate minimizer to $\|\bG\bA\bx-\bG\bb\|_2$ is also a constant-factor approximation to the minimizer of $\|\bS\bA\bR\by-\bS\bb\|_2$ and due to the preconditioner $\bR$, $\kappa(\bS\bA\bR)=\O{1}$ since $\bS$ is also a $(1+\eps)$-approximate subspace embedding of $\bA$. 
Thus, due to the bounded condition number of $\bS\bA\bR$, it suffices to run a small number of steps of gradient descent (GD) to obtain a $(1+\O{\eps})$-approximation to the minimizer of $\|\bS\bA\by-\bS\bb\|_2$ and thus a $(1+\eps)$-approximation to the regression problem $\min_{\bx\in\mathbb{R}^d}\|\bA\bx-\bb\|_2$.

\subsection{Preliminaries}
\label{sec:prelims}
In this paper, we use $[n]$ to denote the set $\{1,\ldots,n\}$ for a positive integer $n$. 
We use $\poly(n)$ to denote a fixed degree polynomial in $n$ that can depend on fixed constants in instantiations of variables throughout the algorithm. 
When a random event occurs with probability at least $1-\frac{1}{\poly(n)}$, we say the event
occurs with high probability. 
Similarly, we use $\polylog(n)$ to denote $\poly(\log n)$. 
We use $\tilde{\Omega}(f)$ to denote $\Omega(f\polylog(f))$ or $\Omega\left(\frac{f}{\polylog(f)}\right)$. 
We use $\mc{N}(\mu,\sigma^2)$ to denote the normal distribution with mean $\mu$ and variance $\sigma^2$ and $\mc{N}(\mathbf{\mu},\bSigma)$ to denote the multivariate normal distribution with mean $\mathbf{\mu}$ and variance $\bSigma$. 

We use the following formulation of the bounded differences inequality:
\begin{definition}[Bounded differences]
For a domain $X$, let $f:X^n\to\mathbb{R}$. 
Then $f$ satisfies the bounded difference assumption if there exist $c_1,\ldots,c_n\ge0$ such that for all $i\in[n]$,
\[\sup_{x_1,\ldots,x_n,x'_i\in X}|f(x_1,\ldots,x_i,\ldots,x_n)-f(x_1,\ldots,x'_i,\ldots,x_n)|\le c_i.\]
\end{definition}

\begin{theorem}[Bounded differences inequality, McDiarmid's inequality]
\cite{mcdiarmid1989method}
\label{thm:mcdiarmid}
Let $X_1,\ldots,X_n\in X$ be independent random variables and suppose $f:x^n\to\mathbb{R}$ satisfies the bounded difference assumption with respect to constants $c_1,\ldots,c_n$. 
Then for all $t>0$,
\[\PPr{f(X_1,\ldots,x_n)-\Ex{f(X_1,\ldots,X_n)}}\le2\exp\left(-\frac{2t^2}{\sum_{i=1}^n c_i^2}\right).\]
\end{theorem}

Given a function $\phi:\mathbb{R}\to\mathbb{R}$, we say the function is $L$-Lipschitz for a parameter $L>0$ if $|f(x)-f(y)|\le L\cdot|x-y|$ for all $x,y\in\mathbb{R}$. 
We use the following formulation of Talagrand's contraction principle:

\begin{theorem}[Ledoux-Talagrand contraction]
\cite{ledtal}
\label{thm:led_tal_cont}
Let $X_1, \dots, X_n \in \R^d$ be i.i.d. random vectors, $\mc{F}$ be a class of real valued functions on $\R^d$ and $\sigma_1, \dots, \sigma_n$ be independent Rademacher random variables. If $\phi: \R \to \R$ is an  $L$-Lipschitz function with $\phi(0) = 0$, then:
\begin{equation*}
\E \sup_{f \in \mc{F}} \sum_{i = 1}^n \sigma_i \phi (f(X_i)) \leq L \cdot \E \sup_{f \in \mc{F}} \sum_{i = 1}^n \sigma_i f(X_i).
\end{equation*}
\end{theorem}

We also use the following bound on the sum of independent mean-zero random variables:
\begin{theorem}[Symmetrization, e.g., Lemma 6.4.2 in~\cite{vershynin2018high}]
\label{thm:symmetrization}
Let $x_1,\ldots,x_n\in\mathbb{R}$ be independent mean-zero random variables. 
Then 
\[\Ex{\norm{\sum_{i=1}^n x_i}_2}\le2\Ex{\norm{\sum_{i=1}^n\sigma_i x_i}_2},\]
where the second expectation is taken over realizations of the random variables $x_i$ and independent Rademacher variables $\sigma_i\in\{-1,+1\}$ for all $i\in[n]$. 
\end{theorem}

We use bold font variables to represent vectors and matrices.  
For a vector $\bv\in\mathbb{R}^n$, we use $\|\bv\|_2$ to denote its Euclidean norm, so that $\|\bv\|_2^2=\sum_{i=0}^n v_i^2$. 
We use $\nnz(\bA)$ to denote the number of nonzero entries in a matrix $\bA\in\mathbb{R}^{n\times d}$ and we use $\bA^{-1}$ to denote the pseudo-inverse of $\bA$. 
For a matrix $\bA\in\mathbb{R}^{n\times d}$, we use \[\|\bA\|_2=\max_{\bx\in\mathbb{R}^d,\|\bx\|_2=1}\|\bA\bx\|_2\]
to denote its operator norm and we use $\kappa(\bA)$ to denote its condition number, so that 
\[\kappa(\bA)=\|\bA\|_2\|\bA^{-1}\|_2.\]
We use $\|\bA\|_F$ to denote the Frobenius norm of a matrix $\bA\in\mathbb{R}^{n\times d}$, so that 
\[\|\bA\|_F^2:=\sum_{i=1}^n\sum_{j=1}^d A_{i,j}^2.\]
For a square matrix $\bM\in\mathbb{R}^{n\times n}$, we use $\mrm{Tr}(\bM)$ to denote its trace, so that $\mrm{Tr}(\bM)=\sum_{i=1}^n M_{i,i}$. 
For a matrix $\bA\in\mathbb{R}^{n\times d}$ and a matrix $\bB\in\mathbb{R}^{n\times p}$, we use $[\bA;\bB]$ to denote the $n\times(d+p)$ dimensional matrix $\begin{bmatrix}\bA & \bB\end{bmatrix}$. 
We use $\bA\succeq 0$ to denote that a matrix $\bA$ is positive semidefinite (PSD). 

We use the following formulation of the Matrix Bernstein inequality:
\begin{theorem}[Matrix Bernstein inequality, e.g., Theorem 1.6 in \cite{Tropp12}]
\label{thm:matrix:bern:ineq}
Let $\bZ_1,\ldots,\bZ_n$ be a sequence of matrices with dimension $n\times d$ such that $\Ex{\bZ_i}=0^{n\times d}$ and $\|\bZ_i\|\le R$ with high probability, for each $i\in[n]$. 
Let $\sigma^2=\max\left(\norm{\sum_{i\in[n]}\Ex{\bZ_i\bZ_i^\top}}_2,\norm{\sum_{i\in[n]}\Ex{\bZ_i^\top\bZ_i}}_2\right)$. 
Then for all $t\ge 0$,
\[\PPr{\norm{\sum_{i\in[n]}\bZ_i}_2\ge t}\le(n+d)\exp\left(-\frac{t^2/2}{\sigma^2+Rt/3}\right).\]
\end{theorem}

\paragraph{Subspace embeddings.}
For an input matrix $\bA\in\mathbb{R}^{n\times d}$, a $(1+\eps)$-subspace embedding for $\bA$ is a matrix $\bM\in\mathbb{R}^{m\times n}$ such that for all $\bx\in\mathbb{R}^{d}$,
\[(1-\eps)\|\bM\bA\bx\|_2\le\|\bA\bx\|_2\le(1+\eps)\|\bM\bA\bx\|_2,\]
for some accuracy parameter $\eps\in(0,1)$. 
The subspace embedding is oblivious if the matrix $\bM$ is generated from a random distribution that is independent of $\bA$; otherwise, the subspace embedding is non-oblivious.

A construction of an oblivious subspace embedding that can be computed in input-sparsity time has the following guarantees:
\begin{theorem}[OSNAP Matrix]
\label{thm:osnap}
\cite{NelsonN13,ClarksonW13,Cohen16}
Given an accuracy parameter $\eps>0$, for any matrix $\bA\in\mathbb{R}^{n\times d}$ with rank $k$, there exists a matrix $\bS\in\mathbb{R}^{m\times n}$ with $m=\O{\frac{1}{\eps^2}k^{1+\alpha}\log k}$ such that for all $\bx\in\mathbb{R}^d$,
\[\|\bA\bx\|_2\le\|\bS\bA\bx\|_2\le(1+\eps)\|\bA\bx\|_2\]
with probability at least $0.99$. 
Moreover, $\bS\bA$ can be computed in time $\O{\frac{\nnz(\bA)}{\alpha\eps}}$.
\end{theorem}

Another construction of oblivious subspace embeddings uses Randomized Hadamard Transforms, which are a family of structured randomized transformations defined as follows:
\begin{definition}[Randomized Hadamard Transform]
\label{def:rht}
\begin{gather*}
    \bH_1 = [1] \qquad \bH_d = 
    \begin{bmatrix}
        \bH_{d / 2} & \bH_{d / 2} \\
        \bH_{d / 2} & -\bH_{d / 2}
    \end{bmatrix} \\
    \forall i \in [m]: \bD^{(i)} \in \mb{R}^{d \times d} \text{ and } D^{(i)}_{j, k} \thicksim 
    \begin{cases}
        \mc{N} (0, 1) & \text{if } j = k \\
        0 & \text{otherwise}
    \end{cases},
    \qquad 
    h(\bz) = 
    \begin{bmatrix}
    \bH\bD^{(1)} \\
    \bH\bD^{(2)} \\
    \vdots \\
    \bH\bD^{(m)} 
    \end{bmatrix} \cdot \bz
\end{gather*}
\end{definition}

We require the following properties of Randomized Hadamard Transforms:
\begin{theorem}{\cite[Theorem 1.1]{CherapanamjeriN22}}
\label{thm:cn_22}
Let $d \in \N, \delta, \varepsilon \in (0, 1/2)$ and $f: \R \to \R$ be a $1$-Lipschitz function. Then for the function $h$ defined in Definition~\ref{def:rht}, we have with probability at least $1 - \delta$:
\begin{equation*}
\forall z \in \R^d \text{ s.t } \norm{z} \leq 1: \abs{\frac{1}{md} \cdot \sum_{i = 1}^{md} f(h (z)_i) - \Exp_{Z \thicksim \mc{N} (0, \norm{z}^2)} [f(z)]} \leq \varepsilon
\end{equation*}
as long as $m \geq C\varepsilon^{-2} \log^5 (d / \varepsilon) \log (1 / \delta)$ for some absolute constant $C > 0$.
\end{theorem}

\begin{lemma}{\cite[Lemma B.6]{CherapanamjeriN22}}
\label{lem:srht_spec_bnd}
For any $d \in \N, \varepsilon, \delta \in (0, 1/2)$, we have that for $m \geq 4 \cdot \frac{\log d + \log (2 / \delta)}{\varepsilon^2}$ and the function $h$ defined in Definition~\ref{def:rht}
\begin{equation*}
\forall \bx \in \R^d: (1 - \varepsilon) \cdot \norm{\bx}_2 \leq \frac{1}{\sqrt{md}} \cdot \norm{h(\bx)}_2 \leq (1 + \varepsilon) \cdot \norm{\bx}_2
\end{equation*}
with probability at least $1-\delta$.
\end{lemma}

\begin{corollary}[Subspace embedding via Randomized Hadamard Transform]
\label{cor:rht}
Given a matrix $\bA\in\mathbb{R}^{n\times d}$ with rank $k$, there exist absolute constants $c,C>0$ and an explicit matrix $\bM\in\mathbb{R}^{m\times n}$ with $m=k\polylog(n)$ rows, such that with probability at least $0.99$, for any vector $\bx\in\mathbb{R}^d$,
\[\frac{1}{2}\|\bA\bx\|_2\le\|\bM\bA\bx\|_2\le\frac{3}{2}\|\bA\bx\|_2\]
and at least $Cm$ of the coordinates of the vector $\bM\bA\bx$ have magnitude at least $\frac{c}{\sqrt{m}}$. 
Moreover, $\bM\bA$ can be computed in time $nd\polylog(n)$. 
\end{corollary}

\paragraph{Leverage scores.}
A non-oblivious construction of a subspace embedding uses the notion of leverage score sampling. 
For a matrix $\bA\in\mathbb{R}^{n\times d}$, the leverage score of row $\ba_i$ with $i\in[n]$ is defined as $\ba_i(\bA^\top\bA)^{-1}\ba_i^\top$. 
Equivalently, for the singular value decomposition $\bA=\bU\bSigma\bV$, the leverage score of row $\ba_i$ is also the squared row norm of $\bu_i$. 
Thus, it is apparent that the sum of the leverage scores of $\bA$ is at most the rank of $\bA$, since the columns of $\bU$ are orthogonal.
\begin{theorem}[Generalization of Foster's Theorem, \cite{foster1953stochastic}]
\label{thm:sum:lev:scores}
For a matrix $\bA\in\mathbb{R}^{n\times d}$, the sum of its leverage scores is $\rank(\bA)$. 
\end{theorem}

It is well-known that leverage score sampling can generate a non-oblivious subspace embedding:
\begin{theorem}[Leverage score sampling]
\label{thm:lev:score:sample}
\cite{DrineasMM06a,DrineasMM06b,MagdonIsmail10,Woodruff14}
Given a matrix $\bA\in\mathbb{R}^{n\times d}$, let $\tau_i$ be the leverage score of the $i$-th row of $\bA$. 
Let $C>1$ be a universal constant and $\alpha>1$ be a parameter and suppose that $p_i\in\left[\min\left(1,\frac{C\tau_i\log k}{\eps^2}\right),\min\left(1,\frac{C\alpha\tau_i\log k}{\eps^2}\right)\right]$ for each $i\in[n]$. 
Let $\bS$ be a random diagonal matrix so that the $i$-th diagonal entry of $\bS$ is $\frac{1}{\sqrt{p_i}}$ with probability $p_i$ and $0$ with probability $1-p_i$. 
Then with probability at least $0.99$, for all vectors $\bx\in\mathbb{R}^d$,
\[(1-\eps)\|\bA\bx\|_2\le\|\bS\bA\bx\|_2\le(1+\eps)\|\bA\bx\|_2.\]
Moreover, $\bS$ has at most $\O{\frac{\alpha}{\eps^2}\,d\log d}$ nonzero entries with probability at least $1-e^{-\Theta(d)}$. 
\end{theorem}
Leverage scores are particularly useful because Theorem \ref{thm:sum:lev:scores} upper bounds the sum of the leverage scores by the rank of the matrix, which is at most $d$ for an input matrix $\bA\in\mathbb{R}^{n\times d}$ with $n\ge d$. 
Thus Theorem \ref{thm:lev:score:sample} implies that only $\O{d\log d}$ rows of $\bA$ need to be sampled for a constant factor subspace embedding of $\bA$, given constant-factor approximations to the leverage scores of $\bA$. 

\paragraph{Gradient descent for linear regression.}
Gradient descent is a well-known iterative method for finding a local minimum of a differentiable function $f$. 
Given a learning rate $\eta>0$ and a point $\bx_n\in\mathbb{R}^n$ for iteration $n$, the point $\bx_{n+1}$ for iteration $n+1$ is defined by
\[\bx_{n+1}=\bx_n+\eta\nabla f(\bx_n),\]
where $\nabla f(\bx_n)$ is the gradient of $f$ at $\bx_n$. 
It is known that the learning rate $\eta$ can be explicitly chosen so that gradient descent achieves the following convergence rate guarantees:
\begin{theorem}[Convergence of gradient descent, e.g., Theorem 3 in \cite{Singer16}]
\label{thm:gd:iterations}
For a convex set $S\subseteq\mathbb{R}^d$, let $f:S\to\mathbb{R}$ be strongly convex on $S$, so that there exist $M>m>0$ such that $mI_d\preceq\nabla^2 f\preceq MI_d$. 
Then for any $\zeta>0$, we have $f(x^{(k)})-\min_{x\in S}f(x)\le\eps$, i.e., $k$ iterations of gradient descent suffice to obtain an additive $\zeta$-approximation, for 
\[k\ge\frac{\log\frac{f(x^{(0)})-y^*}{\zeta}}{\log\frac{M}{M-m}},\]
where $y^*=\min_{x\in S}f(x)$.
\end{theorem}

In our context, we use gradient descent to approximately solve linear regression, i.e., $\min_{\bx\in\mathbb{R}^d}f(\bx):=\min_{\bx\in\mathbb{R}^d}\|\bA\bx-\bb\|_2$ for an input matrix $\bA\in\mathbb{R}^{n\times d}$ and a vector $\bb\in\mathbb{R}^n$. 
This is equivalent to minimizing the squared objective value $\min_{\bx\in\mathbb{R}^d}f(\bx):=\min_{\bx\in\mathbb{R}^d}\|\bA\bx-\bb\|_2^2$. 
In this case, $\nabla f(\bx)=2\bA^\top\bA\bx-2\bA^\top\bb$, which can be explicitly computed from $\bA$ and $\bb$. 
However, if the dimensions of $\bA$ are prohibitively large, it is often desirable to first apply dimensionality reduction techniques to decrease the input size. 

\section{Constant-Factor Subspace Embedding}
\label{sec:constant:se}
In this section, we describe our constant-factor subspace embedding. 
We first require a crude polylogarithmic approximate subspace embedding, which we describe in Section \ref{sec:polylog:se}. 
Our polylogarithmic subspace embedding employs a recently observed property of the SRHT to ``spread'' the mass of an input vector among a large number of coordinates. 
Since a large number of coordinates have a ``large'' amount of mass, it suffices by standard concentration inequalities to uniformly sample rows after the SRHT is applied. 

We then show how to utilize the polylogarithmic subspace embedding to boost the approximation guarantees into a constant-factor subspace embedding in Section \ref{sec:sdp:se}.
Namely, some of the sampled rows in the SRHT could be too large. 
Thus we use fast semidefinite programming to reweight sampled rows of the SRHT to achieve a constant factor approximation. 
We provide more details of this high-level approach in the individual sections.

\subsection{Polylogarithmic Subspace Embedding}
\label{sec:polylog:se}
In this section, we describe our polylogarithmic distortion subspace embedding. 
We first show that after applying the Hadamard Transform, a constant fraction of the resulting vector has ``large'' coordinates.
\begin{lemma}
\label{lem:srht_anti_conc}
There exist constants $c, C > 0$ such that for the function $h$ defined in Definition~\ref{def:rht}:
\begin{equation*}
\forall\bx \in \R^d \text{ s.t } \norm{\bx}_2 = 1: \frac{1}{md} \cdot \sum_{i = 1}^{md} \textbf{1} \lbrb{\abs{h(\bx)_i} \geq c} \geq 0.999 \cdot md
\end{equation*}
with probability at least $1 - \frac{1}{d^{10}}$ as long as $m \geq C \log^6 (d)$.
\end{lemma}
\begin{proof}
Consider the following approximation to an indicator function:
\begin{equation*}
f_c(x) = 
\begin{cases}
0 & \text{if } \abs{x} \leq c \\
\frac{\abs{x} - c}{c} & \text{if } c \leq \abs{x} \leq 2c \\
1 & \text{when } \abs{x} \geq 2c
\end{cases}
\end{equation*}
and let $c$ be small enough such that $\Phi (2c) - \Phi (-2c) \leq 10^{-5}$, where $\Phi$ is the CDF of a standard normal distribution.
Note that such a $c$ exists as the pdf of a standard normal random variable is upper bounded by $1$. 
Now, for large enough $C$ and noting that $f_c$ is a $(1 / c)$-Lipschitz function that $\forall\bx \in \R^d \text{ s.t } \norm{\bx}_2 = 1$, we have by Theorem \ref{thm:cn_22}:
\begin{equation*}
\frac{1}{md} \cdot \sum_{i = 1}^{md} f(h(\bx)_i) \geq \Exp_{Z \thicksim \mc{N} (0, 1)} [f (Z)] - 10^{-5} \geq (1 - (\Phi(c) - \Phi(-c))) - 10^{-5} \geq 1 - 2 \cdot 10^{-5}
\end{equation*}
with probability at least $1 - d^{-10}$. The lemma now follows from the fact that $f(x) \leq \textbf{1} \lbrb{\abs{x} \geq c}$. 
\end{proof}

We now describe our polylogarithmic distortion subspace embedding. 
Given an input matrix $\bA\in\mathbb{R}^{n\times d}$, we first apply OSNAP matrices $\bS_1$ and $\bS_2$ to obtain a constant-factor subspace embedding. 
The OSNAP matrix $\bS_2$ will have sparsity $\alpha=\O{1}$ and thus dimension $\O{d^{1.1}\log d}\times n$ for $\alpha=0.1$, for example. 
This OSNAP matrix will allow us to achieve $\polylog(d)$ dependencies rather than $\polylog(n)$ dependencies. 
Here sparsity $\alpha\in(0,1)$ means that a column of the OSNAP matrix will have $\frac{1}{\alpha\eps}$ nonzero entries, so that an application of the OSNAP matrix incurs runtime proportional to $\frac{1}{\alpha\eps}$.  
Next, the OSNAP matrix $\bS_1$ will have sparsity $\alpha'=\frac{1}{\log d}$ and thus dimension $\O{d\log d}\times\O{d^{1.1}\log d}$. 
The purpose of this OSNAP matrix is to slightly decrease the matrix multiplication time in our analysis, though we remark that without $\bS_1$, our analysis can still be performed but simply requiring a smaller value of $\alpha$ for $\bS_2$. 

Our polylogarithmic subspace embedding is simple. 
We apply an SRHT matrix $\bM$ with $d\,\polylog(d)$ rows to $\bS_1\bS_2\bA$. 
From standard results on Hadamard Transforms, $\bM\bS_1\bS_2\bA$ is actually a constant-factor subspace embedding for $\bA$. 
In light of Lemma \ref{lem:srht_anti_conc}, $\bM\bS_1\bS_2\bA\bx$ has a ``large'' number of coordinates that are ``large'', for any unit vector $\bA\bx\in\mathbb{R}^d$. 
Thus we can uniformly sample rows of $\bM\bS_1\bS_2\bA$ and achieve a ``good'' approximation to $\|\bA\bx\|_2$. 
Hence, our polylogarithmic subspace embedding is simply the matrix $\bS\bM\bS_1\bS_2\bA$, where $\bS$ is a matrix that uniformly samples $\O{d}$ rows independently. 

\begin{theorem}
\label{thm:polylog:se}
For any $\bA \in \R^{n \times d}$ and a tradeoff parameter $\alpha > 0$, we may compute matrix $\bG \in \R^{p \times n}$ such that:
\[\forall \bx \in \R^d: \norm{\bA\bx}_2 \leq \norm{\bG\bA\bx}_2 \leq \polylog(d)\norm{\bA\bx}_2,\]
with probability at least $0.9$ for some constant $\xi > 0$. Furthermore, we have $p = \O{d}$ and $\bG\bA$ may be computed in time $\O{\frac{\nnz(\bA)}{\alpha}} + d^{2 + \alpha} \polylog (d)$.
\end{theorem}
\begin{proof}
Let $\bS_1$ and $\bS_2$ be OSNAP matrices (Theorem \ref{thm:osnap}) that induce a $2$-approximate subspace embedding, such that $\alpha=0.1$ for $\bS_2$ and $\alpha=\frac{1}{\log d}$ for $\bS_1$. Consider two successive applications of OSNAP matrices $\bS_1$ and $\bS_2$ to $\bA$ to obtain $\bB = \bS_1 \bS_2 \bA \in \R^{\ell\times d}$, where $\ell=\O{d\log d}$ by Theorem \ref{thm:osnap}. Note that by Theorem \ref{thm:osnap}, $\bS_2\bA$ can be computed in time $\O{\nnz(\bA)}$ and thus subsequently, $\bS_1\bS_2\bA$ can be computed in time $\O{d^{2.1}\log d}$. Moreover, the exponent $2.1$ is due to the choice $\alpha=0.1$ for $\bS_2$ and can be made any arbitrary constant greater than $2$. 

Next, consider an SRHT, characterized by matrix $\bM$ with $m = \polylog(\ell)=\polylog(d)$ rows. 
For the corresponding linear mapping $h$, we have:
\begin{gather*}
\forall \bx \in \R^\ell: (1 - \varepsilon) \norm{\bx}_2 \leq \frac{1}{\sqrt{m\ell}} \cdot \norm{h(\bx)}_2 \leq (1 + \varepsilon) \cdot \norm{\bx}_2 \\
\forall \bx \in \R^\ell \text{ s.t } \norm{\bx}_2 = 1: \frac{1}{m\ell} \cdot \sum_{i \in [m\ell]} \textbf{1} \lbrb{\abs{h(\bx)_i} \geq c} \geq 0.999\\
\forall i \in [m\ell]: \norm{\bbm_i}_2 \leq 2 \sqrt{\ell} \tag{SRHT-COND} \label{eq:srht_cond}
\end{gather*}
with probability at least $1 - 1 /\ell^{10}$ by Lemma \ref{lem:srht_spec_bnd}, Lemma \ref{lem:srht_anti_conc}, and the fact that the lengths of the rows of $\bM$ correspond to the length of one of $m$ independently distributed standard normal random vectors. 
Here, we use $\textbf{1}$ to denote an indicator variable and $\bbm_i$ to denote the $i$-th row of $\bM$. 

Next, consider a subsampling matrix, $\bS \in \R^{p \times m\ell}$ where each row of $\bS$ is uniformly sampled from the set of elementary vectors $\{\be_i\}_{i \in [m\ell]}$ for $p = Cd$ for some suitably large constant $C$. 
We now have by an application of the matrix Bernstein inequality, i.e., Theorem \ref{thm:matrix:bern:ineq}:
\begin{equation*}
\frac{1}{\sqrt{Cd}} \cdot \norm{\bS\bM}_2 \leq \polylog (d)
\end{equation*}
with probability at least $1 - 1 /\ell^{10}$. 
Now letting $T$ be the multiset of indices selected in the construction of $\bS$, consider the random variable:
\begin{equation*}
Z = \sup_{\bu \in \mathrm{Span} (\bB) \text{ s.t } \norm{\bu}_2 = 1} \abs{\frac{1}{p} \cdot \sum_{i \in T} \textbf{1} \lbrb{\inp{\bbm_i}{\bu} \geq c} - \Exp_{i \in [m\ell]} [\textbf{1} \lbrb{\inp{\bbm_i}{\bu} \geq c}]}.
\end{equation*}
Since $Z$ satisfies the bounded differences assumption with respect to the elements of $T$ and the fact that $Z$ corresponds to the empirical concentration of indicator functions of halfspaces of dimension $d$, we have by standard VC Theory and McDiarmid's inequality, i.e., Theorem \ref{thm:mcdiarmid}, that $Z \leq 0.001$ with probability at least $1 - 1 /\ell^{10}$. 
From this, we get that for a suitably large $C > 0$:
\begin{equation*}
\forall \bu \in \mrm{Span} (\bB): \frac{1}{\sqrt{p}} \cdot \norm{\bS\bM\bu}_2 \geq c \cdot \norm{\bu}_2 \implies \norm{\bu}_2 \leq \frac{C}{\sqrt{p}} \norm{\bS\bM\bu}_2 \leq \polylog (d) \cdot \norm{\bu}_2.
\end{equation*}
Hence, for the matrix $\bG\bA=\bS\bM\bS_1\bS_2\bA$, we have that $\bG\bA$ is a subspace embedding with $\polylog(d)$-distortion that can be computed in time $\O{\nnz(\bA)+d^{2.1}\log d}$. 
However, as we previously remarked, the exponent $2.1$ is due to the choice $\alpha=0.1$ for $\bS_2$ and can be made any arbitrary constant greater than $2$ with the tradeoff that the $\nnz(\bA)$ term becomes $\frac{\nnz(\bA)}{\alpha}$ in the application of the OSNAP matrix $\bS_2$ in Theorem \ref{thm:osnap}. 
Therefore, the overall runtime is $\O{\frac{\nnz(\bA)}{\alpha}} + d^{2 + \alpha} \polylog (d)$.
\end{proof}

\subsection{Constant-Factor Subspace Embedding}
\label{sec:sdp:se}
In this section, we describe how to improve our polylogarithmic factor subspace embedding into a constant-factor subspace embedding. 
We first require the following guarantees for (approximately) solving semidefinite programs (SDPs).

\begin{theorem}[Theorem 1.1 in \cite{PengTZ12}]
\label{thm:ptz:oneone}
For a primal positive SDP with $m\times m$ matrices and $n$ constraints and an accuracy parameter $\eps>0$, there exists an algorithm that produces a $(1+\eps)$-approximation in $\O{\frac{1}{\eps^3}\log^3 n}$ iterations, where each iteration consists of computing matrix sums and a special primitive that computes $\exp(\Phi)\bullet\bA$ for positive semidefinite matrices (PSD) $\Phi$ and $\bA$. 
\end{theorem}
Here, $\exp(\Phi)\bullet\bA$ denotes the pointwise dot product between matrices $\exp(\Phi)$ and $\bA$, i.e., the Hadamard product. 
\begin{theorem}[Theorem 4.1 in \cite{PengTZ12}]
\label{thm:ptz:fourone}
There exists an algorithm that takes input an $m\times m$ matrix $\Phi$ with $p$ nonzero entries, $\kappa\ge\max(1,\|\Phi\|_2)$, and PSD $m\times m$ matrices $\bA_i$ in factorized form $\bA_i=\bQ_i\bQ_i^\top$, where the total number of nonzero entries across all matrices $\bQ_i$ is $q$, and outputs a  $(1+\eps)$-approximation to all $\exp(\Phi)\bullet\bA_i$. 
The algorithm uses $\O{\frac{1}{\eps^2}\left(p\kappa\log\frac{1}{\eps}+q\right)\log m}$ total work. 
\end{theorem}

We now describe our constant-factor subspace embedding. 
Recall that our polylogarithmic distortion subspace embedding for an input matrix $\bA\in\mathbb{R}^{n\times d}$ is a matrix $\bS\bM\bS_1\bS_2\bA$, where $\bS$ is a matrix that randomly and independently samples $\O{d}$ rows, $\bM$ is an SRHT matrix with $d\,\polylog(d)$ rows and $\bS_1$ and $\bS_2$ are OSNAP matrices. 
Recall furthermore that from standard results for randomized Hadamard Transforms, $\bM\bS_1\bS_2\bA$ is already a constant-factor subspace embedding. 
The reason $\bS\bM\bS_1\bS_2\bA$ is not a constant-factor subspace embedding is because there are certain rows of $\bM\bS_1\bS_2\bA$ that are too large. 

We first show that with high probability, there exists a reweighting $\bW$ of the sampled rows so that $\bW\bS\bM\bS_1\bS_2\bA$ is a constant-factor subspace embedding of $\bA$. 
We can thus use semidefinite programming to efficiently compute such a set of weights and quickly output $\bW\bS\bM\bS_1\bS_2\bA$. 
A high-level description of our constant-factor subspace embedding is summarized in Figure \ref{fig:constant:se}. 

\thmse*
\begin{proof}
Let $\bS_1\in\mathbb{R}^{\O{\ell}\times\tO {d^{1 + \alpha}}}$ and $\bS_2\in\mathbb{R}^{\tO{d^{1 + \alpha}}\times n}$ be the OSNAP matrices defined in Theorem \ref{thm:polylog:se}, so that each matrix is a constant-factor subspace embedding and $\alpha = 0.1$ for $\bS_2$ and $\alpha = \frac{1}{\log d}$ for $\bS_1$. 
In particular, we have $\ell=\O{d\log d}$. 
Let $\bM$ be an SRHT with $m = \polylog(\ell)=\polylog(d)$ rows, as in Theorem \ref{thm:polylog:se}. 
Finally, we let $\bS \in \R^{p \times m\ell}$ be a subsampling matrix, where each row of $\bS$ is uniformly sampled from the set of elementary vectors $\{\be_i\}_{i \in [m\ell]}$ for $p = Cd$ for some suitably large constant $C$, as in Theorem \ref{thm:polylog:se}. Note that as in the proof of Theorem~\ref{thm:polylog:se}, we may set $\alpha$ for $\bS_2$ to any constant $\alpha > 0$ to obtain the final result.
    
Let $\bx_1, \dots, \bx_p$ denote the rows of $\bS\bM$, $\bB = \bS_1 \bS_2 \bA$ and $\bU$ be an orthonormal basis for $\mathrm{Span} (\bB)$. 
We will now find a set of weights $w \in \mc{W} \coloneqq \lbrb{w: \sum_{i \in [p]} w_i = 1 \text{ and } \forall i \in [p],\, 0 \leq w_i \leq \frac{2}{p}} \subset \R^p$ such that $\norm{\bU\bU^\top \cdot \bSigma_w \cdot \bU\bU^\top}_2$ is minimized where $\bSigma_w = \sum_{i \in [p]} w_i \bx_i \bx_i^\top$. We start by showing that this quantity is small with high probability. 
    
\begin{lemma}
\label{lem:const_fact_se_ub}
We have for some suitably large constant $C > 0$:
\begin{equation*}
\min_{w \in \mc{W}} \norm{\bU \bU^\top \cdot \bSigma_w \cdot \bU \bU^\top}_2 \leq C
\end{equation*}
With probability at least $1 - 1 /p^{10}$.
\end{lemma}
\begin{proof}
We start by analyzing the quantity via the approach from \cite{lm19}:
\begin{equation*}
Z_r = \sup_{\bu \in \bU \text{ s.t } \norm{\bu}_2 = 1} \frac{1}{p} \cdot \sum_{i = 1}^p \textbf{1} \lbrb{\abs{\inp{\bu}{\bx_i}} \geq r}.
\end{equation*}
Note that $Z_r$ satisfies the bounded differences inequality with respect to the rows $x_i$ drawn i.i.d. from the rows of $\bM$. Furthermore, we have:
\begin{align*}
\Exp [Z_r] &\leq \frac{1}{r} \cdot \Exp \lsrs{\sup_{\bu \in \bU \text{ s.t } \norm{\bu}_2 = 1} \frac{1}{p} \cdot \sum_{i \in [p]} \abs{\inp{\bu}{\bx_i}}} \\
&\leq \frac{1}{r} \lprp{\Exp \lsrs{\sup_{\bu \in \bU \text{ s.t } \norm{\bu}_2 = 1} \frac{1}{p} \cdot \sum_{i \in [p]} \lprp{\abs{\inp{\bu}{\bx_i}} - \Exp_{\bx \in \bM} \abs{\inp{\bu}{\bx}}}} + \sup_{\bu \in \bU \text{ s.t } \norm{\bu}_2 = 1} \Exp_{\bx \in \bU} [\abs{\inp{\bu}{\bx}}]}
\end{align*}
Since $\bu$ and $\bx$ are both unit vectors, then $\Exp_{\bx \in \bU} [\abs{\inp{\bu}{\bx}}]\le 1$. 
Thus we have
\begin{align*}
\Exp [Z_r] &\leq \frac{1}{r} \lprp{\Exp_{\bx_i, \bx_i'} \lsrs{\sup_{\bu \in \bU \text{ s.t } \norm{\bu}_2 = 1} \frac{1}{p} \cdot \sum_{i \in [p]} \lprp{\abs{\inp{\bu}{\bx_i}} - \abs{\inp{\bu}{\bx_i'}}}} + 1},
\end{align*}
Note that $\abs{\inp{\bu}{\bx_i}} - \abs{\inp{\bu}{\bx_i'}}$ is a zero-mean random variable. 
Thus by using standard symmetrization arguments, i.e., Theorem \ref{thm:symmetrization}, we can insert Radamacher variables $\sigma_i$ as follows where $x_i'$ represent independent copies of $x_i$. 
Therefore,
\begin{align*}
\Exp [Z_r]&\leq \frac{2}{r} \lprp{\Exp_{\bx_i, \bx_i', \sigma_i} \lsrs{\sup_{\bu \in \bU \text{ s.t } \norm{\bu}_2 = 1} \frac{1}{p} \cdot \sum_{i \in [p]} \sigma_i \lprp{\abs{\inp{\bu}{\bx_i}} - \abs{\inp{\bu}{\bx_i'}}}} + 1} \\
&\leq \frac{2}{r} \lprp{\Exp_{\bx_i, \bx_i', \sigma_i} \lsrs{\sup_{\bu \in \bU \text{ s.t } \norm{\bu}_2 = 1} \frac{1}{p} \cdot \sum_{i \in [p]} \sigma_i \abs{\inp{\bu}{\bx_i}} + \sup_{\bu \in \bU \text{ s.t } \norm{\bu}_2 = 1} \frac{1}{p} \cdot \sum_{i \in [p]} - \sigma_i \abs{\inp{\bu}{\bx_i'}}} + 1} \\
&= \frac{4}{r} \lprp{\Exp_{\bx_i, \sigma_i} \lsrs{\sup_{\bu \in \bU \text{ s.t } \norm{\bu}_2 = 1} \frac{1}{p} \cdot \sum_{i \in [p]} \sigma_i \abs{\inp{\bu}{\bx_i}}} + 1}.
\end{align*}
Since the Rademacher random variables $\sigma_i \in \{\pm 1\}$ are chosen uniformly at random and independent of the $x_i$, we can remove the absolute values around $\langle u, x_i \rangle$ using Talagrand's contraction principle, i.e., Theorem \ref{thm:led_tal_cont}. 
Hence, by the definition of the operator norm,
\begin{align*}
\Exp [Z_r]&\leq \frac{4}{r} \lprp{\Exp_{\bx_i, \sigma_i} \lsrs{\sup_{\bu \in \bU \text{ s.t } \norm{\bu}_2 = 1} \frac{1}{p} \cdot \sum_{i \in [p]} \sigma_i \inp{\bu}{\bx_i}} + 1} \\
&= \frac{4}{r} \lprp{\Exp_{\bx_i, \sigma_i} \lsrs{\norm{\bU^\top \lprp{\frac{1}{p} \cdot \sum_{i \in [p]} \sigma_i \bx_i}}_2} + 1}.
\end{align*}
By convexity and Jensen's inequality, we have
\begin{align*}
\Exp [Z_r]&\leq \frac{4}{r} \lprp{\lprp{\Exp_{\bx_i, \sigma_i} \lsrs{\norm{\bU^\top \lprp{\frac{1}{p} \cdot \sum_{i \in [p]} \sigma_i \bx_i}}_2^2}}^{1/2} + 1} \\
&\leq \frac{4}{r} \lprp{\lprp{\Exp_{\bx_i, \sigma_i} \lsrs{
\mrm{Tr}\lprp{\bU^\top\left(\lprp{\frac{1}{p} \cdot \sum_{i \in [p]} \sigma_i \bx_i}\right)\left(\lprp{\frac{1}{p} \cdot \sum_{i \in [p]} \sigma_i \bx_i}^\top\right)\bU}}}^{1/2}+1}\\
&\leq \frac{4}{r} \cdot \lprp{\lprp{\mrm{Tr} \lprp{\bU^\top \lprp{\frac{1}{p} \cdot \Exp_{\bx \in \bM} [\bx\bx^\top]} \bU}}^{1/2} + 1},
\end{align*}
where the last inequality follows from the linearity of the trace operator. 
Since $\bU$ is an orthonormal basis for $\mathrm{Span}(\bB)$, a subspace of dimension $d$ and $\bM$ satisfies $\Exp_{x \in \bM} [xx^\top] \preccurlyeq 2 \cdot I$ (\ref{eq:srht_cond}):
\begin{align*}
\Exp [Z_r]\leq \frac{4}{r} \cdot \lprp{\lprp{\frac{2d}{p}}^{1/2} + 1}.
\end{align*}
Hence, we get for large enough $r$ that $\Exp [Z_r] \leq 0.0001$. 
The bounded differences inequality, i.e., Theorem \ref{thm:mcdiarmid},  now yields that $Z_r \leq 0.0002$ with probability at least $1 - 1 / \ell^{10}$. 
Now, we analyze the random variable:
\begin{equation*}
\min_{\bw \in \mc{W}} \norm{\bU\bU^\top \bSigma_w \bU\bU^\top}_2 = \min_{\bw \in \mc{W}} \max_{\bY \succcurlyeq 0, \mrm{Tr}(\bY) = 1} \inp{\bU\bU^\top \bSigma_w \bU\bU^\top}{\bY} = \max_{\bY \succcurlyeq 0, \mrm{Tr}(\bY) = 1} \min_{\bw \in \mc{W}} \inp{\bU\bU^\top \bSigma_w \bU\bU^\top}{\bY}
\end{equation*}
where the exchange of the min and max follows from von Neumann's equality. We will now show that for all $\bY \succcurlyeq 0, \mrm{Tr}(\bY) = 1$, we have $T_{\bY} \coloneqq \abs{\lbrb{i: \bx_i^\top \bU\bU^\top \bY \bU\bU^\top \bx_i \geq 32768 r^2}} < 0.5 p$ using an analysis similar to \cite{dl22}.
Suppose for the sake of contradiction, there exists such an $\bY$ satisfying this. Then, consider a Gaussian random variable $g \thicksim \mc{N} (0, \bY)$. We now have for all $i \in T_{\bY}$ by noting that $\bg^\top \bU\bU^\top \bx_i$ is a zero mean gaussian random variable with variance $\Exp [(\bg^\top \bU\bU^\top \bx_i)^2] = \Exp [\bx_i^\top \bU\bU^\top \bg\bg^\top \bU\bU^\top \bx_i] = \inp{\bY}{\bU\bU^\top \bx_i \bx_i^\top \bU\bU^\top} \geq 32768r^2$ and standard upper bounds on the pdf of a gaussian random variable:
\begin{gather*}
\P \lbrb{\abs{\bg^\top \bU\bU^\top \bx_i} \geq 8r} \geq \frac{9}{10} \\
\P \lbrb{\norm{\bg}_2 \leq 4} \geq \frac{9}{10}.
\end{gather*}
Hence, we get by a union bound on the above two events:
\begin{equation*}
\P \lbrb{\frac{1}{\norm{\bg}_2} \cdot \abs{\bg^\top \bU\bU^\top \bx_i} \geq 2r} \geq \frac{8}{10}.
\end{equation*}
And we get:
\begin{equation*}
p \cdot Z_r \geq \Exp_{\bg} \lsrs{\sum_{i \in T_{\bY}} \textbf{1} \lbrb{\frac{1}{\norm{\bg}_2} \cdot \abs{\bg^\top \bU\bU^\top \bx_i} \geq 2r}} \geq \frac{8}{10} \cdot \abs{T_{\bY}} \geq 0.4p,
\end{equation*}
which yields a contradiction and establishes the lemma. 
\end{proof}
    
The task of finding a suitable set of weights can be formulated as the following packing semi-definite program:
\begin{align*}
\max\ &1^\top \bw \\
\text{s.t } &\sum_{i \in [p]} w_i \bA_i \preccurlyeq C \cdot \bI \\
&w_i \geq 0 \\
&\text{where } \bA_i = \boC_i \boC_i^\top \text{ with }
\boC_i = 
\begin{bmatrix}
\bU\bU^\top \cdot \bx_i & 0 \\
0 & \sqrt{\frac{2}{p}} \cdot \be_i 
\end{bmatrix}.
\end{align*}
These families of SDPs may be solved to constant accuracy in time $d^2 \polylog (d)$ from Theorem \ref{thm:ptz:oneone} by noting that the matrix exponentials computed in Theorem \ref{thm:ptz:fourone} may be implemented with the fast projection oracle onto the span of $\bU$ via gradient descent to inverse polynomial accuracy. 
    
Towards concluding the proof, let $\bW$ denote the diagonal matrix with $W_{i,i} = \sqrt{w_i}$. We now have $\norm{\bW\bS \bU\bU^\top}_2 \leq C$ for some $C > 0$ from the constraints of the program. 

\begin{lemma}
    \label{lem:const_fact_se_lb}
    We have for some absolute $c > 0$:
    \begin{equation*}
        \forall \bw \in \mc{W}, \bu \in \bU \text{ s.t } \norm{\bu}_2 = 1: \bu^\top \lprp{\sum_{i \in [p]} w_i \cdot \bx_i \bx_i^\top} \bu \geq c
    \end{equation*}
    with probability at least $1 - 1 / p^{10}$.
\end{lemma}
\begin{proof}
    We have from \ref{eq:srht_cond}, that for any $\bu \in \bU$ with $\norm{\bu}_2 = 1$ for some $c > 0$:
    \begin{equation*}
        Q(\bu) \coloneqq \frac{1}{ml} \sum_{i \in [ml]} \textbf{1} \lbrb{\abs{h(\bu)_i} \geq c} \geq 0.999
    \end{equation*}
    Noting that the VC-dimension of halfspaces of dimension $d$ is $d + 1$, we have by \cite[Theorem 8.3.23]{vershynin2018high}:
    \begin{equation*}
        \Exp \lsrs{\underbrace{\sup_{\bu \in \bU, \norm{\bu}_2 = 1} \abs{\frac{1}{p} \cdot \sum_{i \in [p]} \textbf{1} \lbrb{\abs{\inp{\bx_i}{\bu}} \geq c} - Q(\bu)}}_{Q}} \leq C \cdot \sqrt{\frac{d}{p}}
    \end{equation*}
    for some absolute constant $C > 0$. Furthermore, noting that the random variable $Q$ satisfies the bounded differences inequality as the rows $\bx_i$ are drawn i.i.d. from the rows of $\bM$. Hence, we have by an application of the bounded differences inequality and the previous display that with probability at least $1 - 1 / p^{10}$:
    \begin{equation*}
        \forall \bu \in \bU \text{ s.t } \norm{\bu}_2 = 1: \frac{1}{p} \cdot \sum_{i \in [p]} \textbf{1} \lbrb{\abs{\inp{\bx_i}{\bu}} \geq c} \geq 0.99.
    \end{equation*}
    The lemma now follows from the fact that for any $\bw \in \mc{W}$, from the fact that $w_i \leq 2 / p$:
    \begin{equation*}
        \norm{\bW \bS \bM \bu}_2 = \sqrt{\sum_{i \in [p]} w_i \inp{\bx_i}{\bu}^2} \geq c \cdot \sqrt{\sum_{i \in [p]} w_i \textbf{1} \lbrb{\abs{\inp{\bx_i}{\bu}} \geq c}} \geq \frac{c}{2}.
    \end{equation*}
\end{proof}

To conclude the proof, note that $\bS_1 \bS_2$ is a valid constant-factor subspace embedding for $\bA$; that is, there exist constants $C^\prime, \xi^\prime > 0$ such that:
\begin{equation*}
    \forall \bx \in \R^d: \norm{\bA \bx}_2 \leq C^\prime \norm{\bS_1 \bS_2 \bA \bx}_2 \leq \xi^\prime \norm{\bA \bx}_2.
\end{equation*}
Furthermore, we have as a consequence of Lemmas~\ref{lem:const_fact_se_ub} and \ref{lem:const_fact_se_lb} that there exist constants $C'', \xi''$:
\begin{align*}
    \forall \bu \in \bU: \norm{\bu}_2 \leq C'' \norm{\bW \bS \bM \bu}_2 &= \sqrt{\bu^\top \cdot \lprp{\sum_{i \in [p]} w_i \bx_i\bx_i^\top} \cdot \bu} \\
    &= \sqrt{\bu^\top \bU\bU^\top \Sigma_{\bw} \bU\bU^\top \bu} \leq \sqrt{\norm{\bU\bU^\top\Sigma_\bw \bU\bU^\top}_2} \cdot \norm{\bu}_2 \leq \xi'' \norm{\bu}_2
\end{align*}
Recalling that $\bU$ is the span of $\bS_1 \bS_2 \bA$, the previous two displays yield:
\begin{equation*}
\forall \norm{\bx}_2 = 1: \norm{\bA\bx}_2 \leq C \norm{\bW\bS\bM\bS_1\bS_2\bA\bx}_2 \leq \xi \norm{\bA\bx}_2
\end{equation*}
for some constants $C,\xi > 0$ with probability at least $1 - 1 / d^{10}$.
    
\paragraph{Runtime analysis.} We start by showing that we may approximately project onto $\bU$ by computing a pre-conditioner, $\bR$, of $\bB$ such that $\bB \bR$ has condition number $\polylog (d)$:
\begin{lemma}
    We may compute in time $\O{d^{\omega}}$ a matrix $\bR$ such that $\bB \bR$ has condition number $\polylog(d)$.
\end{lemma}
\begin{proof}
    We start by computing a pre-conditioner of $\bS \bM \bS_1 \bS_2 \bA$, $\bR$, in time $\O{d^\omega}$ as $\bS\bM\bS_1\bS_2\bA$ is of dimension $\O{d} \times d$. Consequently, $\bS \bM \bS_1 \bS_2 \bA \bR$ has condition number $\polylog (d)$. The lemma will now follow by showing that $\bS \bM$ is a $\polylog(d)$-subspace embedding for $\bU$. The lower bound follows from Lemma~\ref{lem:const_fact_se_lb}. The upper bound now follows from applying matrix Bernstein (i.e.,  Theorem~\ref{thm:matrix:bern:ineq}) to the random matrix $\sum_{i \in [p]} \bx_i \bx_i^\top$ by noting that $\norm{\bx_i}_2 \leq \sqrt{d} \polylog (d)$ and $\Exp [(\bx_i\bx_i^\top)^2] \preccurlyeq 4 \ell \cdot \bI \preccurlyeq 4 d \polylog (d) \cdot \bI$ by \ref{eq:srht_cond}. 
\end{proof}

As a consequence of the above lemma, we may compute an approximate projection onto $\bU$ in time $d^2 \polylog (d) \log (1 / \gamma)$ with accuracy $\gamma$; i.e., for any $\norm{\bu}_2 = 1$, we can compute a vector $\hat{\bu} \in \bU$ such that $\norm{\hat{\bu} - \bu^*} \leq \gamma$ where $\bu^* = \argmin_{\bz \in \bU} \norm{\bu - \bz}$ via gradient descent. Now, to determine the runtime, we first note that by Theorem \ref{thm:osnap}, $\bS_1\bA$ can be computed in time $\O{\nnz(\bA)}$ and has dimension $d^{1.1}\log d\times d$. 
Similarly by Theorem \ref{thm:osnap}, $\bS_2\bS_1\bA$ can be subsequently computed in time $\O{d^{2.1}\log^2(d)}$ and has dimension $d\log d\times d$. 
By Corollary \ref{cor:rht}, $\bM\bS_2\bS_1\bA$ can then subsequently be computed in time $d^2\polylog(d)$ and has dimensions $d\polylog(d)\times d$. 
Since $\bS$ is a sampling matrix that samples $\O{d}$ rows, then $\bS\bM\bS_2\bS_1\bA$ can be subsequently computed in time $\O{d^2}$ and has dimensions $\O{d}\times d$. 
Since the SDP can be solved to constant accuracy in time $d^2\polylog(d)$, then $\bW$ can be computed in time $d^2\polylog(d)$. 
Since $\bW$ is simply a reweighting matrix, then $\bW\bS\bM\bS_2\bS_1\bA$ can subsequently be computed in time $\O{d^2}$. 
Therefore, the total time to compute the subspace embedding $\bW\bS\bM\bS_2\bS_1\bA$ is 
\[\O{\nnz(\bA)+d^{2.1}\log^2(d)}+d^2\polylog(d).\]

More generally, we can use an arbitrary $\alpha$ instead of setting $\alpha=0.1$ to achieve the total runtime \[\O{\frac{\nnz(\bA)}{\alpha}+d^{2+\alpha}\log^2(d)}+d^2\polylog(d)=\O{\frac{\nnz(\bA)}{\alpha}} + d^{2 + \alpha} \polylog(d).\]
\end{proof}

\section{Subspace Embedding through Leverage Score Sampling}
In this section, we show that our constant factor approximation can be used to achieve leverage score sampling in the current matrix-multiplication runtime. 
Leverage score sampling is an important tool that will allow us to achieve a $(1+\eps)$-subspace embedding in this section, approximate linear regression in Section \ref{sec:lin:reg}, and  independent row selection in Section \ref{sec:ind:row}. 

We first recall the following standard result, which states that a constant-factor subspace embedding can be used to achieve constant-factor approximations to the leverage scores. 
\begin{lemma}
\label{lem:qr:lev}
Suppose $\bS$ is a subspace embedding for $\bA\in\mathbb{R}^{n\times d}$ so that for any $\bx\in\mathbb{R}^d$,
\[\|\bA\bx\|_2\le\|\bS\bA\bx\|_2\le\alpha\|\bA\bx\|_2.\]
Then for all $i\in[n]$,
\[\frac{\tau}{\alpha^2}\le\|\ba_i\bR\|_2^2\le\tau,\]
where $\bS\bA=\bQ\bR^{-1}$ for an orthonormal matrix $\bQ$, i.e., $\bQ\bR^{-1}$ is the QR decomposition of $\bS\bA$, and $\tau_i$ is the leverage score of the $i$-th row of $\bA$. 
\end{lemma}
Lemma \ref{lem:qr:lev} follows from the fact that $\bQ$ has orthonormal columns and that the leverage score of each row of $\bA$ is just the squared row norm of $\bU$ in the singular value decomposition $\bA=\bU\bSigma\bV$, see, e.g., Section \ref{sec:prelims}, while $\bS$ is an $\alpha$-distortion subspace embedding. 
More detailed proofs of Lemma \ref{lem:qr:lev} appear in \cite{DrineasMMW12,Woodruff14,ChepurkoCKW22}. 

To quickly obtain a constant factor approximation to $\|\ba_i\bR\|_2^2$ for all $i\in[n]$, a standard approach is to use a Gaussian matrix $\bG$ with $\O{\log n}$ columns and then compute $\|\ba_i\bR\bG\|_2^2$~\cite{DrineasMMW12,Woodruff14}. 
However, this multiplication by a dense matrix incurs a high runtime. 
Instead, \cite{ChepurkoCKW22} showed that a two-stage sampling process can be performed by first using a Gaussian matrix $\bG'$ with only $\O{1/\gamma}$ columns, so that $\|\ba_i\bR\bG'\|_2^2$ is an $\O{n^\gamma\log n}$-approximation to $\|\ba_i\bR\|_2^2$ for each $i\in[n]$.  
Then after sampling a large number of rows, i.e., oversampling by an $\O{n^\gamma\log n}$ factor, we can compute constant-factor approximations to the sampled rows and then perform rejection sampling to reduce the overall number of rows. 
Formally, the guarantees are as follows:

\begin{theorem}[Lemma 7.3 in~\cite{ChepurkoCKW22}]
\label{thm:produce:sample}
Given $\bA\in\mathbb{R}^{n\times d}$, suppose $\bR\in\mathbb{R}^{d\times d}$ is a matrix such that for any vector $\bx\in\mathbb{R}^d$, $\bA\bR\bx$ can be computed in time $T_1$ and $\bR\bx$ can be computed in time $T_2$. 
Given parameters $\alpha,s>0$, there exists an algorithm that with probability at least $0.95$, samples a random subset $S\subseteq[n]$ such that each $i\in S$ with probability $f_i$, where
\[\min\left(1,\frac{s}{16}\frac{\|\ba_i\bR\|_2^2}{\|\bA\bR\|_F^2}\right)\le f_i\le\min\left(1,s\,\frac{\|\ba_i\bR\|_2^2}{\|\bA\bR\|_F^2}\right).\]
The algorithm outputs $S$ along with $f_i$ for each $i\in S$ in time $\O{\frac{T_1}{\alpha}+T_2\log n+sdn^\alpha\log^2(n)}$. 
\end{theorem}

Given these previous results, our algorithm is simple. 
To perform leverage score sampling on an input matrix $\bA$ to obtain a $(1+\eps)$-subspace embedding, we first obtain a constant-factor subspace embedding $\bS_1\bA$ through Theorem \ref{thm:se}. 
We then compute a QR decomposition of $\bS_1\bA$ so that $\bQ\bR^{-1}=\bS_1\bA$ and apply Theorem \ref{thm:produce:sample} to $\bA$ and $\bR$. 
The algorithm in full appears in Algorithm \ref{alg:eps:se}. 

\begin{algorithm}[!htb]
\caption{Subspace embedding through leverage score sampling}
\label{alg:eps:se}
\begin{algorithmic}[1]
\Require{$\bA\in\mathbb{R}^{n\times d}$, $\eps,\alpha>0$}
\Ensure{Subspace embedding $\bS\bA$}
\State{Let $\bS_1\bA$ be a fast embedding of $\bA$}
\Comment{Theorem \ref{thm:se}}
\State{Let $\bQ\bR^{-1}$ be a QR decomposition of $\bS_1\bA$}
\State{$s\gets\frac{1}{\eps^2}\,d\polylog(d)$}
\State{Let $S$ and $\{f_i\}$ be the output of Theorem \ref{thm:produce:sample} with inputs $\bA,\bR,s,\alpha$}
\State{Set the $i$-th diagonal entry of $\bS=\frac{1}{\sqrt{f_i}}$ for each $i\in S$}
\State{\Return $\bS\bA$}
\end{algorithmic}
\end{algorithm}

We now show that Algorithm \ref{alg:eps:se} can be used to obtain a $(1+\eps)$-subspace embedding in the current matrix-multiplication runtime. 
\thmalgepsse*
\begin{proof}
By Theorem \ref{thm:se}, there exists a matrix $\bS_1$ such that with probability at least $0.99$, $\bS_1$ has $\O{d}$ rows and 
\[\|\bA\bx\|_2\le\|\bS_1\bA\bx\|_2\le\xi\|\bA\bx\|_2,\]
for $\xi=\O{1}$ and for all $\bx\in\mathbb{R}^d$. 
Moreover, $\bS_1\bA$ can be computed in $\O{\frac{\nnz(\bA)}{\alpha}}+d^{2+\alpha}\polylog(d)$ time for any $\alpha>0$. 
Thus the QR decomposition of $\bS_1\bA$ can be computed in time $\O{d^\omega}$ to output matrices $\bQ$ and $\bR^{-1}$ such that $\bQ$ has orthonormal columns and $\bQ\bR^{-1}=\bS_1\bA$. 
By Lemma \ref{lem:qr:lev}, 
\[\frac{\tau_i}{\xi^2}\le\|\ba_i\bR\|_2^2\le\tau_i,\]
for all $i\in[n]$, so that
\[\frac{\tau_i}{d\xi^2}\le\frac{\|\ba_i\bR\|_2^2}{\|\bA\bR\|_F^2},\]
since $\|\bA\bR\|_F^2=\sum_{i\in[n]}\|\ba_i\bR\|_2^2\le\sum_{i\in[n]}\tau_i\le d$ by Theorem \ref{thm:sum:lev:scores}. 
By Theorem \ref{thm:produce:sample}, there exists an algorithm that with probability at least $0.95$, will output a set $S$ along with corresponding sampling probabilities $f_i$, for each $i\in S$, such that
\[f_i\ge\min\left(1,\frac{s}{16}\frac{\|\ba_i\bR\|_2^2}{\|\bA\bR\|_F^2}\right)\ge\min\left(1,\frac{s}{16}\frac{\tau_i}{d\xi^2}\right).\]
Setting $s=\frac{1}{\eps^2}\,d\xi^2$, it follows that $f_i\ge\min\left(1,\frac{C\tau_i\log d}{\eps^2}\right)$ for some constant $C>0$.  
Thus, by Theorem \ref{thm:lev:score:sample}, with probability at least $0.9$ for the matrix $\bS$ of Algorithm \ref{alg:eps:se},
\[(1-\eps)\|\bA\bx\|_2\le\|\bS\bA\bx\|_2\le(1+\eps)\|\bA\bx\|_2.\]
By setting $T_1=\nnz(\bA)+d^2$ and $T_2=d^2$ in Theorem \ref{thm:produce:sample}, it follows that the total runtime is
\[\O{\frac{\nnz(\bA)}{\alpha}+d^\omega}+\frac{1}{\eps^2}\,n^\alpha d^2\polylog(d).\]
Now we note that either $d>n^{0.1}$, in which case $n^{\alpha}$ can be replaced with $d^{\alpha}$ after a reparametrization of $\alpha$ or $d<n^{0.1}$, in which case the $n^{\alpha} d^2\polylog(k)$ term is lower-order since $\nnz(\bA)$ can be assumed to be at least $n$ by throwing out zero rows. 
Therefore, the final runtime is \[\O{\frac{\nnz(\bA)}{\alpha}+d^\omega}+\frac{1}{\eps^2}\,d^{2+\alpha}\polylog(d).\]
\end{proof}

\section{Independent Row Selection}
\label{sec:ind:row}
In this section, we show how our leverage score sampling framework and our constant-factor subspace embedding can be used to select a maximal set of independent rows of an input matrix $\bA$ in the current matrix-multiplication runtime. 
We first require the following definition of rank-preserving sketches:
\begin{definition}[Rank-preserving sketch]
A distribution $\calS$ on matrices $\bS\in\mathbb{R}^{m\times n}$ is a rank preserving sketch if there exists a constant $c>0$ such that for $\bS\sim\calS$, with high probability, for an input matrix $\bA\in\mathbb{R}^{n\times d}$, we have $\min\left(\rank(\bS\bA),\frac{m}{c}\right)=\min\left(\rank(\bA),\frac{m}{c}\right)$. 
\end{definition}

\cite{CheungKL13} gave a sparse construction of a rank-preserving sketch with the following properties:

\begin{theorem}
\cite{CheungKL13}
\label{thm:rps}
There exists a rank-preserving sketch distribution with $c=11$ such that (1) $\bS\bA$ can be computed in $\O{\nnz(\bA)}$ time, (2) $\bS$ has at most two nonzero entries in a column, and (3) $\bS$ has at most $\frac{2n}{m}$ nonzeros in a row. 
\end{theorem}

We also require the following guarantees from the approximate matrix multiplication algorithm. 
More precisely, the approximate matrix multiplication algorithm samples a fixed number of rows with replacement from an input matrix $\bM$, where each row $\bbm_i$ is sampled with probability proportional to $\|\bbm_i\|_2^2$. 
The result is a matrix $\bS\bM$ such that $\bM^\top\bS^\top\bS\bM$ is a ``good'' approximation to $\bM^\top\bM$. 

\begin{theorem}
\label{thm:amm}
\cite{DrineasKM06}
Suppose $\bS$ is a sampling matrix with $r$ rows randomly generated from the approximate matrix multiplication algorithm on input $\bM\in\mathbb{R}^{n\times d}$. 
Then with probability at least $\frac{2}{3}$,
\[\|\bM^\top\bS^\top\bS\bM-\bM^\top\bM\|_F^2\le\frac{10}{\sqrt{r}}\,\|\bM\|_F^4.\]
\end{theorem}

The approximate matrix multiplication algorithm is simply squared row norm sampling, which is equivalent to leverage score sampling after preconditioning, since the leverage scores of each row of $\bA$ are just the squared row norms of $\bU$ in the singular value decomposition $\bA=\bU\bSigma\bV$, e.g., see Section \ref{sec:prelims}.

\subsection{Independent Row Selection for a Reduced Matrix}
\label{sec:basis:grow}
We now describe Algorithm \ref{alg:ind:row}, our algorithm for independent row selection from a matrix $\bB\in\mathbb{R}^{\O{k\log(k)}\times\O{k}}$ with rank $k$ -- we shall ultimately reduce the input matrix $\bA\in\mathbb{R}^{n\times d}$ with rank $k$ down to this case. 
Our algorithm will iteratively grow a set $S$ of independent rows of $\bB$. 

We first use leverage score sampling to sample $\O{k}$ rows of $\bB$. 
Although we require $\O{k\log k}$ samples to cover the entire row span of $\bB$, we can show using the guarantees of approximate matrix multiplication by Theorem \ref{thm:amm} that $\O{k}$ samples suffice to cover a span with rank $(1-c)k=\frac{9}{10}k$ with probability $\frac{2}{3}$, for $c=\frac{1}{10}$. 
We can efficiently compute both an independent subset of these rows and a basis $\bZ_1$ for their orthogonal complement. 
We can then add the independent subset to $S$. 
These procedures in total use at most $\gamma k^\omega$ runtime for some fixed constant $\gamma>0$.  

For the next iteration, we repeat these procedures on $\bB\bZ_1^\top$. 
Namely, we use leverage score sampling to sample $\O{k}$ rows of $\bB\bZ_1^\top$. 
We can again efficiently compute both an independent subset of these rows and a basis $\bZ_2$ for their orthogonal complement. 
We can again add an independent subset of these rows to our growing set $S$ and then compute a basis $\bZ_2$ for the orthogonal complement of $S$. 
The main idea is that the rows of $\bB$ that are spanned by $S$ will have leverage score zero in $\bB\bZ_1^\top$ because they must be orthogonal to $\bZ_1$, the orthogonal complement of $S$. 
Hence, the sampled rows will cover a constant fraction of the remaining subspace orthogonal to $S$ and in particular with probability $\frac{2}{3}$, we can sample a set of rows with rank at least a $\frac{9}{10}$ fraction of the rank of $\bB\bZ_1^\top$. 

Moreover, since the rank of $\bB\bZ_1^\top$ is at most a $c=\frac{1}{10}$-fraction of the rank of $\bB$ conditioned on the success of the first iteration, then the runtime of the second iteration is at most $\gamma(ck)^\omega$, which is a constant fraction smaller than the runtime of the first iteration. 
We can now repeatedly apply this approach by iteratively adding a set of rows to $S$ that cover a constant fraction of the remaining subspace orthogonal to $S$, while using runtime a constant fraction of that in the previous iterations. 
Since the runtime follows a geometric series in expectation, the overall runtime is $\O{k^\omega}$ in expectation, and so by Markov's inequality, the overall runtime is $\O{k^\omega}$ with constant probability. 
We remark that the correctness and runtime analysis is robust to failures in each iteration because a failure in an iteration means that at worst, no additional rows are added to $S$, which does not affect the correctness of the algorithm and only slightly increases the runtime, which is absorbed into the computation of the expected runtime.  

\begin{algorithm}[!htb]
\caption{Independent Row Selection}
\label{alg:ind:row}
\begin{algorithmic}[1]
\Require{$\bA\in\mathbb{R}^{m\times d}$ for $m=d\,\polylog(d)$}
\Ensure{A set of $\rank(\bA)$ independent rows of $\bA$}
\State{Use Theorem \ref{thm:lev:score:sample} so that $S$ is a set of $\O{\rank(\bA)}$ rows of $\bA$}
\State{Let $\bS\bA$ be the submatrix of $\O{\rank(\bA)}$ rows of $\bA$ in $S$}
\State{Let $\bZ^{(1)}$ be a basis for the orthogonal complement of $\bS\bA$}
\State{Reduce $S$ to a set of independent rows}
\State{$t\gets 1$}
\While{$\bZ^{(t)}$ is non-empty}
\State{Use Theorem \ref{thm:lev:score:sample} so that $S'$ is a set of $\O{r}$ rows of $\bA\bZ^{(1)}\ldots\bZ^{(t)}$, where $r=\rank(\bA\bZ^{(1)}\ldots\bZ^{(t)})$}
\State{$S\gets S\cup S'$}
\State{Reduce $S$ to a set of independent rows}
\State{Let $\bS\bA$ be the submatrix of $\O{r}$ rows of $\bA$ in $S$}
\State{Let $\bZ^{(t+1)}$ be a basis for the orthogonal complement of $\bS\bA$}
\EndWhile
\State{\Return $S$}
\end{algorithmic}
\end{algorithm}

We show that the subroutine Algorithm \ref{alg:ind:row} can be used to find a set of $\rank(\bA)$ independent rows of $\bA$ in matrix-multiplication runtime. 
We shall ultimately apply Algorithm \ref{alg:ind:row} to a matrix $\bA\in\mathbb{R}^{\O{k\log k}\times\O{k}}$. 
\begin{lemma}
\label{lem:smaller:ind}
Given a matrix $\bA\in\mathbb{R}^{n\times d}$ of rank $k$, there exists an algorithm that with probability at least $\frac{2}{3}$, outputs a set $S$ of $k$ independent rows of $\bA$ in time 
\[\O{\nnz(\bA)\log d+k^\omega}.\]
\end{lemma}
\begin{proof}
Suppose $\bA$ has rank $k$ and let $S$ be a set of $n_1\ge\frac{9}{10}k$ rows of $\bA$, with rank $r_1$.  
Let $\bS$ be the corresponding sampling matrix so that $S$ consists of the rows of $\bS\bA\in\mathbb{R}^{n_1\times d}$. 
Let $\bZ^{(1)}$ be a basis for the orthogonal complement of $\bS\bA$, so that $\bZ^{(1)}\in\mathbb{R}^{(k-n_1)\times d}$. 

Observe that for any row $\ba_i$ in the span of $S$, we have $\ba_i(\bZ^{(1)})^\top=0^{(k-n_1)}$, where the right-hand side denotes the all zeros vector of length $k-n_1$. 
Thus the only nonzero rows of $\bA(\bZ^{(1)})^\top$ are the rows that are independent of $S$ and so we would like to sample rows of $\bA$ proportional to their leverage score sample in $\bA(\bZ^{(1)})^\top$. 
However, due to our desired runtime, we cannot afford to explicitly compute $\bA(\bZ^{(1)})^\top$. 
Instead we apply the techniques of Theorem \ref{thm:alg:eps:se} to perform leverage score sampling on $\bA(\bZ^{(1)})^\top$. 
Namely, we first generate a matrix $\bG^{(1)}\in\mathbb{R}^{m_1\times n_1}$ via Theorem \ref{thm:se} and suppose that
\[\|\bG^{(1)}\bA(\bZ^{(1)})^\top\bx\|_2\le\|\bG^{(1)}\bS\bA(\bZ^{(1)})^\top\bx\|_2\le\gamma\|\bG^{(1)}\bA(\bZ^{(1)})^\top\bx\|_2,\]
for an absolute constant $\gamma>1$ and for all $\bx\in\mathbb{R}^d$. 
Crucially, $\bG^{(1)}\bA(\bZ^{(1)})^\top\in\mathbb{R}^{m_1\times d}$, where $m_1\le C(k-n_1)$ and can be computed in time $C(\nnz(\bA)+m_1k^{\omega-1})$ for an absolute constant $C>0$. 
As in Theorem \ref{thm:alg:eps:se}, we then use $C(\nnz(\bA)+m_1k^{\omega-1})$ time to compute a QR decomposition of $\bG^{(1)}\bA(\bZ^{(1)})^\top$ to output matrices $\bQ$ and $\bR$ such that $\bQ$ has orthonormal columns and $\bQ\bR=\bG^{(1)}\bA(\bZ^{(1)})^\top$. 

By Lemma \ref{lem:qr:lev}, we have that $\|\ba_i(\bZ^{(1)})^\top\bR^{-1}\|_2^2$ is a $\xi^2$-approximation to the leverage score of the $i$-th row of $\ba_i(\bZ^{(1)})^\top$. 
Thus by Theorem \ref{thm:produce:sample}, we can sample $\O{d-n_1}$ rows of $\bA$ with probabilities proportional to the leverage scores of $\bA(\bZ^{(1)})^\top$. 

By setting $\bM$ in the context of Theorem \ref{thm:amm} to $\bM=\bG^{(1)}\bA(\bZ^{(1)})^\top\bR^{-1}$, we have that $\bM^\top\bM$ is the diagonal matrix consisting of $r_1:=d-n_1$ ones and zeros elsewhere, since $\bQ$ has orthonormal columns. 
Thus, $\|\bM\|_F^2=r_1$ and so by Theorem \ref{thm:amm}, we have that for $r=\O{r_1}$ with probability at least $\frac{2}{3}$, 
\[\|\bM^\top\bS^\top\bS\bM-\bM^\top\bM\|_F^2\le\frac{r_1}{100},\]
where $\bS$ is the sampling matrix induced by approximate matrix multiplication, which is equivalent to leverage score sampling in this case. 
On the other hand, we have $\|\bM^\top\bS^\top\bS\bM-\bM^\top\bM\|_F^2\ge\rank(\bS\bM)$, since $\bM^\top\bM$ is a diagonal matrix consisting of only ones and zeros. 
Therefore, it follows that $\bS\bM$, i.e., the set of rows from leverage score sampling, has found at least a $\frac{9}{10}$ fraction of the remaining independent rows. 

By arguing inductively, the algorithm outputs a set of $d$ independent rows. 
Namely, we use $\bS$ to compute the matrix $\bZ^{(2)}$ for the orthogonal complement of the sampled rows. 
Then given a sequence of matrices $\bZ^{(1)},\ldots,\bZ^{(i)}$, we generate a matrix $\bG^{(i)}$ via Theorem \ref{thm:se} and its QR decomposition to iteratively perform leverage score sampling. 

Call an iteration successful if the number of remaining independent rows of $\bA$ has decreased by at least a  $\frac{9}{10}$ fraction. 
We define a round to be a number of iterations such that the number of remaining independent rows of $\bA$ has decreased by at least a $\frac{9}{10}$ fraction. 
Since each iteration in round $t$ runs in time 
\[C\left(\nnz(\bA)+\left(\frac{1}{10}\right)^{i-1}k\cdot k^{\omega-1}\right)\]
and succeeds with probability at least $\frac{2}{3}$, then the expected runtime $R_t$ of round $t$ is at most
\[\Ex{R_t}\le C\left(\nnz(\bA)+\left(\frac{1}{10}\right)^{i-1}k\cdot k^{\omega-1}\right)+\frac{1}{3}\Ex{R_t}\le2C\left(\nnz(\bA)+\left(\frac{1}{10}\right)^{i-1}k\cdot k^{\omega-1}\right).\]
Thus the total expected runtime is at most
\[\sum_{i=1}^{\log d}=2C\left(\nnz(\bA)+\left(\frac{1}{10}\right)^{i-1}k\cdot k^{\omega-1}\right)=\O{\nnz(\bA)\log d+k^\omega}.\]
Hence we have by Markov's inequality that with probability at least $\frac{2}{3}$, the algorithm uses total time
\[\O{\nnz(\bA)\log d+k^\omega}.\]
\end{proof}

\subsection{Input Matrix Reduction}
We now show there exists an algorithm for independent row selection that uses the current matrix-multiplication runtime. 
We would like to reduce from an input matrix $\bA\in\mathbb{R}^{n\times d}$ with rank $k$ to a matrix $\bB\in\mathbb{R}^{\O{k\log k}\times k}$, which would allow us to apply Algorithm \ref{alg:ind:row} and therefore, Lemma \ref{lem:smaller:ind}. 
To that end, we apply a rank-preserving sketch $\bS$ to $\bA$, so that $\rank(\bA)=\rank(\bA\bS)$, where $\bS\in\mathbb{R}^{d\times ck}$ for some constant $c>0$ and any set of independent rows of $\bA$ is also a set of independent rows of $\bA\bS$.  
We then use our constant-factor subspace embedding to select $\O{k\log k}$ reweighted rows from $\bA\bS$. 
These reweighted rows form the input matrix $\bB$ to Algorithm \ref{alg:ind:row}. 

\thmindrows*
\begin{proof}
Let $\bS_1\in\mathbb{R}^{ck\times d}$ be a rank-preserving sketch and let $\calE$ be the event that $\rank(\bA\bS_1^\top)=\rank(\bA)=k$. 
Then by Theorem \ref{thm:rps}, we have $\PPr{\calE}\ge1-\frac{1}{\poly(n)}$. 

Conditioned on $\calE$, let $I\subseteq[n]$ be a subset of independent rows of $\bA\bS_1^\top$, so that $\rank(\bA_I)=k$. 
Thus to find $k$ linearly independent rows of $\bA$, it suffices to find $k$ linearly independent rows of $\bA\bS_1^\top$. 
Let $\bB=\bA\bS_1^\top$, so that $\nnz(\bB)=\O{\nnz(\bA)}$ by Theorem \ref{thm:rps}. 

By Theorem \ref{thm:alg:eps:se}, there exists an algorithm that samples $\O{k\log k}$ rows of $\bB$ to form a matrix $\bB'$ such that $\rank(\bB)=\rank(\bB')=k$, using time \[\O{\frac{\nnz(\bA)}{\alpha}+k^\omega}+k^{2+\alpha}\polylog(k),\]
for any trade-off parameter $\alpha>0$. 
Thus, we can then apply Lemma \ref{lem:smaller:ind} to $\bB'$ to compute $k$ linearly independent rows in time $k^2\,\polylog(k)+\O{k^\omega}$.
Thus, the overall runtime is 
\[\O{\frac{\nnz(\bA)}{\alpha}+k^\omega}+k^{2+\alpha}\polylog(k).\]
\end{proof}

\section{Linear Regression}
\label{sec:lin:reg}
In this section, we show how our $(1+\eps)$-subspace embedding can be used to solve approximate linear regression in the current matrix-multiplication runtime.  

We first recall the following statement that shows how a $(1+\sqrt{\eps})$-subspace embedding suffices to achieve a $(1+\O{\eps})$-approximate solution to linear regression. 
\begin{lemma}[Theorem 14 in \cite{BourgainDN15}]
\label{lem:se:lr:solve}
Let $\bS\bA$ be a $(1+\sqrt{\eps})$-subspace embedding of an input matrix $\bA\in\mathbb{R}^{n\times d}$. 
Then
\[\min_{\bx\in\mathbb{R}^d}\|\bS\bA-\bS\bb\|_2\le(1+\O{\eps})\cdot\min_{\bx\in\mathbb{R}^d}\|\bA-\bb\|_2.\]
\end{lemma}
Thus Lemma \ref{lem:se:lr:solve} implies that we should simply find an approximate solution to linear regression after applying a $(1+\sqrt{\eps})$-subspace embedding, e.g., Theorem \ref{thm:alg:eps:se}. 
However, this is not straightforward because the resulting dimension after applying Theorem \ref{thm:alg:eps:se} would be $\O{\frac{1}{\eps}\,d\log d}$, which is not small enough to compute the closed-form solution for linear regression in the current matrix-multiplication runtime, due to the extra logarithmic factor. 
Instead, we consider gradient descent to find an approximately optimal solution for linear regression, recalling the guarantees on the convergence rate of gradient descent in Theorem \ref{thm:gd:iterations}. 

We now show that there exists an algorithm for approximate linear regression in the current matrix-multiplication runtime. 
The main approach is similar to that of \cite{ChepurkoCKW22}, but we have a better runtime due to our constant-factor subspace embedding, and also we have a better dependence on $\eps$ due to invoking Lemma \ref{lem:se:lr:solve}. 

The main idea is that Lemma \ref{lem:se:lr:solve} states that it suffices to solve approximate linear regression after applying a $(1+\O{\sqrt{\eps}})$-subspace embedding $\bS$, which by Theorem \ref{thm:alg:eps:se} results in $\O{\frac{1}{\eps}\,d\log d}$ rows. 
Unfortunately, the dimension of $\bS$ is too high to find a closed form solution to $\min_{\bx\in\mathbb{R}^d}\|\bS\bA\bx-\bS\bb\|_2$. 
On the other hand, since we only require finding a vector $\bw\in\mathbb{R}^d$ such that $\|\bS\bA\bw-\bS\bb\|_2\le(1+\O{\eps})\min_{\bx\in\mathbb{R}^d}\|\bS\bA\bx-\bS\bb\|_2$, we instead use gradient descent to find such a vector $\bw$. 
However, gradient descent requires a small condition number and a ``good'' initial solution. 
To decrease the condition number to $\O{1}$, we instead consider $\min_{\bx\in\mathbb{R}^d}\|\bS\bA\bR\bx-\bS\bb\|_2$, where $\bG\bA=\bQ\bR^{-1}$ is a QR decomposition for a constant-factor subspace embedding $\bG\bA$. 
To find a good initial solution, we first find the closed-form solution to $\bw^{(0)}=\argmin_{\bx\in\mathbb{R}^d}\|\bG\bA\bx-\bG\bb\|_2$, since $\bG\bA$ is a constant-factor subspace embedding. 
We then account for the preconditioning by computing $\bw^{(1)}=\bR^{-1}\bw^{(0)}$, which is a good starting point for our gradient descent because it provides a constant-factor approximation to the optimal solution due to the properties of $\bG\bA$. 

\thmlinreg*
\begin{proof}
Let $\alpha>0$ be a fixed constant. 
By Theorem \ref{thm:alg:eps:se}, we can in time $\O{\frac{\nnz(\bA)}{\alpha}+d^\omega}+\frac{1}{\eps}\,d^{2+\alpha}\polylog(d)$, compute a $(1+\O{\sqrt{\eps}})$ subspace embedding $[\bS\bA; \bS\bb]$ of the matrix $[\bA; \bb]$ with probability at least $0.9$.  
By Theorem \ref{thm:se}, we can compute a matrix $\bG$ such that with probability at least $0.9$, $\bG$ has $\O{d}$ rows and $\bG\bA$ is a $\xi$-distortion subspace embedding with $\xi=\O{1}$. 
Moreover, since $\bG\bA\in\mathbb{R}^{\O{d}\times d}$, then we can compute its QR decomposition $\bG\bA=\bQ\bR^{-1}$ in $\O{d^\omega}$ time. 
Because $\bQ$ has orthonormal columns, then the condition number of $\bG\bA\bR$ is $\kappa(\bG\bA\bR)=1$. 
Since $\bG\bA$ is a $\xi$-distortion subspace embedding with $\xi=\O{1}$, it follows that $\kappa(\bA\bR)=1$. 
Similarly, we have $\kappa(\bS\bA\bR)=\O{1}$ since $\bS$ is also a $(1+\O{\sqrt{\eps}})$ subspace embedding of $\bA$. 
Intuitively, $\bR$ serves as a good preconditioner to the matrix $\bA$. 

More precisely, let $\bw$ be an approximate minimizer of the resulting matrix, so that \[\|\bS\bA\bR\bw-\bS\bb\|_2\le(1+\eps)\min_{\bx\in\mathbb{R}^d}\|\bS\bA\bR\bx-\bS\bb\|_2.\]
Then by Lemma \ref{lem:se:lr:solve}, $\bR\bw$ is a $(1+\O{\eps})$-approximate solution to the linear regression problem, so it remains to compute $\bR\bw$. 

Unfortunately, since $\bS\bA\bR$ has $\O{\frac{1}{\eps}\,d\log d}$ rows, we cannot afford to immediately use the closed-form solution to compute $\min_{\bx\in\mathbb{R}^d}\|\bS\bA\bR\bx-\bS\bb\|_2$. 
On the other hand, since $\bR\bw$ is a $(1+\O{\eps})$-approximate solution to the linear regression problem, we can use gradient descent to compute $\bR\bw$ after finding a ``good'' initial point. 

To that end, we first find the closed-form solution to $\bw^{(0)}=\argmin_{\bx\in\mathbb{R}^d}\|\bG\bA\bx-\bG\bb\|_2$, since $\bG\bA$ is a constant-factor subspace embedding. 
We then account for the preconditioning by computing $\bw^{(1)}=\bR^{-1}\bw^{(0)}$, which will serve as a starting point for our gradient descent. 

More specifically, let $\bw^{(0)}=(\bG\bA)^+(\bG\bb)$, where $(\bG\bA)^+$ is the pseudo-inverse of $\bG\bA$. 
Since $\bG\bA\in\mathbb{R}^{\O{d}\times d}$, then we can compute its pseudo-inverse in $\O{d^\omega}$ time. 
Moreover, since the construction of $\bG$ in Theorem \ref{thm:se} consists of a reweighted subsampled Hadamard transform, then we can compute $\bG\bb$ in $\O{n\log n}=\O{\nnz(\bA)}$ time. 
Thus we can compute $\bw^{(0)}$ in total time $\O{\nnz(\bA)+d^\omega}$ after computing $\bG$.  

We can now compute $\bw^{(1)}=\bR^{-1}\bw^{(0)}$ in $\O{d^2}$ time and furthermore,
\[\|\bS\bA\bR\bw^{(1)}-\bS\bb\|_2\le(1+\eps)\|\bA\bR\bw^{(1)}-\bb\|_2=(1+\eps)\|\bA\bw^{(0)}-\bb\|_2\le(1+\eps)\|\bG\bA\bw^{(0)}-\bG\bb\|_2.\]
Let $\bz=\argmin_{\bx\in\mathbb{R}^d}\|\bS\bA\bx-\bS\bb\|_2$. 
Since $\bw^{(0)}=(\bG\bA)^+(\bG\bb)$, then $\bw^{(0)}=\argmin_{\bx\in\mathbb{R}^d}\|\bG\bA\bx-\bG\bb\|_2$.
Therefore,
\begin{align*}
\|\bS\bA\bR\bw^{(1)}-\bS\bb\|_2&\le(1+\eps)\|\bG\bA\bw^{(0)}-\bG\bb\|_2\\
&\le(1+\eps)\|\bG\bA\bz-\bG\bb\|_2\\
&\le(1+\eps)\gamma\|\bA\bz-\bb\|_2\\
&\le(1+\eps)\gamma\min_{\bx\in\mathbb{R}^d}\|\bS\bA\bx-\bS\bb\|_2.
\end{align*}

In other words, $\bw^{(1)}$ is an $\O{1}$-approximation to the optimizer of the linear regression problem for the input matrix $\bS\bA\bR$ and the measurement vector $\bS\bb$, since $\gamma=\O{1}$. 
Since $\bQ$ is orthonormal, then (squared) linear regression for $[\bS\bA\bR; \bS\bb]$ is $1$-strongly convex. 
Moreover, since $\kappa(\bS\bA\bR)=\O{1}$, then we can set $m=1$ and $M=\O{1}$ in Theorem \ref{thm:gd:iterations}. 
Further, setting the parameters  $f(x^{(0)})\le\gamma\min_{\bx\in\mathbb{R}^d}\|\bS\bA\bx-\bS\bb\|_2$ and the gradient descent accuracy $\zeta=\eps\min_{\bx\in\mathbb{R}^d}\|\bS\bA\bx-\bS\bb\|_2$ in Theorem \ref{thm:gd:iterations}, we obtain a $(1+\eps)$-approximation by using $\O{\log\frac{1}{\eps}}$ iterations of gradient descent with the initial solution as $\bw^{(1)}$. 
Since $\bS$ has $\O{\frac{1}{\eps}\,d\log(d)}$ rows from Theorem \ref{thm:alg:eps:se}, each iteration of gradient descent can be performed in time $\frac{1}{\eps}\,d^2\polylog d$. 
Hence, the overall runtime to compute a $(1+\eps)$-approximate solution to the linear regression problem on input matrix $\bA$ and measurement vector $\bb$ is 
\[\O{\frac{\nnz(\bA)}{\alpha}+d^\omega}+\frac{1}{\eps}\,d^{2+\alpha}\polylog(d)+\frac{1}{\eps}\,d^2\polylog(d)\log\frac{1}{\eps}.\] 
\end{proof}

\paragraph{Acknowledgments.} 
We thank Jelani Nelson for helpful discussions over the course of the project. 
David P. Woodruff and Samson Zhou were supported by a Simons Investigator Award and by the National Science Foundation under Grant No. CCF-1815840. Sandeep Silwal is supported by an NSF Graduate Research Fellowship under Grant No.\ 1745302, and NSF TRIPODS program (award DMS-2022448), NSF award CCF-2006798, and Simons Investigator Award (via Piotr Indyk).

\bibliographystyle{alpha}
\bibliography{references}

\newcommand{\etalchar}[1]{$^{#1}$}
\begin{thebibliography}{DMMW12}

\bibitem[AW21]{AlmanW21}
Josh Alman and Virginia~Vassilevska Williams.
\newblock A refined laser method and faster matrix multiplication.
\newblock In {\em Proceedings of the {ACM-SIAM} Symposium on Discrete
  Algorithms, {SODA}}, pages 522--539, 2021.

\bibitem[BDN15]{BourgainDN15}
Jean Bourgain, Sjoerd Dirksen, and Jelani Nelson.
\newblock Toward a unified theory of sparse dimensionality reduction in
  euclidean space.
\newblock In {\em Proceedings of the Forty-Seventh Annual {ACM} on Symposium on
  Theory of Computing, {STOC}}, pages 499--508. {ACM}, 2015.

\bibitem[CCKW22]{ChepurkoCKW22}
Nadiia Chepurko, Kenneth~L. Clarkson, Praneeth Kacham, and David~P. Woodruff.
\newblock Near-optimal algorithms for linear algebra in the current matrix
  multiplication time.
\newblock In {\em Proceedings of the {ACM-SIAM} Symposium on Discrete
  Algorithms, {SODA}}, pages 3043--3068, 2022.

\bibitem[CKL13]{CheungKL13}
Ho~Yee Cheung, Tsz~Chiu Kwok, and Lap~Chi Lau.
\newblock Fast matrix rank algorithms and applications.
\newblock {\em J. {ACM}}, 60(5):31:1--31:25, 2013.

\bibitem[CN22]{CherapanamjeriN22}
Yeshwanth Cherapanamjeri and Jelani Nelson.
\newblock Uniform approximations for randomized hadamard transforms with
  applications.
\newblock In {\em {STOC}: 54th Annual {ACM} {SIGACT} Symposium on Theory of
  Computing}, pages 659--671, 2022.

\bibitem[Coh16]{Cohen16}
Michael~B. Cohen.
\newblock Nearly tight oblivious subspace embeddings by trace inequalities.
\newblock In {\em Proceedings of the Twenty-Seventh Annual {ACM-SIAM} Symposium
  on Discrete Algorithms, {SODA}}, pages 278--287, 2016.

\bibitem[CW82]{CoppersmithW82}
Don Coppersmith and Shmuel Winograd.
\newblock On the asymptotic complexity of matrix multiplication.
\newblock {\em {SIAM} J. Comput.}, 11(3):472--492, 1982.

\bibitem[CW13]{ClarksonW13}
Kenneth~L. Clarkson and David~P. Woodruff.
\newblock Low rank approximation and regression in input sparsity time.
\newblock In {\em Symposium on Theory of Computing Conference, STOC}, pages
  81--90, 2013.

\bibitem[DKM06]{DrineasKM06}
Petros Drineas, Ravi Kannan, and Michael~W. Mahoney.
\newblock Fast monte carlo algorithms for matrices {I:} approximating matrix
  multiplication.
\newblock {\em {SIAM} J. Comput.}, 36(1):132--157, 2006.

\bibitem[DL22]{dl22}
Jules Depersin and Guillaume Lecu\'{e}.
\newblock Robust sub-{G}aussian estimation of a mean vector in nearly linear
  time.
\newblock {\em Ann. Statist.}, 50(1):511--536, 2022.

\bibitem[DMM06a]{DrineasMM06a}
Petros Drineas, Michael~W. Mahoney, and S.~Muthukrishnan.
\newblock Subspace sampling and relative-error matrix approximation:
  Column-based methods.
\newblock In {\em Approximation, Randomization, and Combinatorial Optimization.
  Algorithms and Techniques, 9th International Workshop on Approximation
  Algorithms for Combinatorial Optimization Problems, {APPROX} and 10th
  International Workshop on Randomization and Computation, {RANDOM},
  Proceedings}, pages 316--326, 2006.

\bibitem[DMM06b]{DrineasMM06b}
Petros Drineas, Michael~W. Mahoney, and S.~Muthukrishnan.
\newblock Subspace sampling and relative-error matrix approximation:
  Column-row-based methods.
\newblock In {\em Algorithms - {ESA} 2006, 14th Annual European Symposium,
  Proceedings}, pages 304--314, 2006.

\bibitem[DMMW12]{DrineasMMW12}
Petros Drineas, Malik Magdon{-}Ismail, Michael~W. Mahoney, and David~P.
  Woodruff.
\newblock Fast approximation of matrix coherence and statistical leverage.
\newblock {\em J. Mach. Learn. Res.}, 13:3475--3506, 2012.

\bibitem[Fos53]{foster1953stochastic}
Frederic~G Foster.
\newblock On the stochastic matrices associated with certain queuing processes.
\newblock {\em The Annals of Mathematical Statistics}, 24(3):355--360, 1953.

\bibitem[Ind07]{Indyk07}
Piotr Indyk.
\newblock Uncertainty principles, extractors, and explicit embeddings of l2
  into l1.
\newblock In {\em Proceedings of the 39th Annual {ACM} Symposium on Theory of
  Computing}, pages 615--620, 2007.

\bibitem[LM19]{lm19}
G\'{a}bor Lugosi and Shahar Mendelson.
\newblock Sub-{G}aussian estimators of the mean of a random vector.
\newblock {\em Ann. Statist.}, 47(2):783--794, 2019.

\bibitem[LS17]{LeeS17}
Yin~Tat Lee and He~Sun.
\newblock An sdp-based algorithm for linear-sized spectral sparsification.
\newblock In {\em Proceedings of the 49th Annual {ACM} {SIGACT} Symposium on
  Theory of Computing, {STOC}}, pages 678--687, 2017.

\bibitem[LS18]{LeeS18}
Yin~Tat Lee and He~Sun.
\newblock Constructing linear-sized spectral sparsification in almost-linear
  time.
\newblock {\em {SIAM} J. Comput.}, 47(6):2315--2336, 2018.

\bibitem[LT91]{ledtal}
Michel Ledoux and Michel Talagrand.
\newblock {\em Probability in {B}anach spaces}, volume~23 of {\em Ergebnisse
  der Mathematik und ihrer Grenzgebiete (3) [Results in Mathematics and Related
  Areas (3)]}.
\newblock Springer-Verlag, Berlin, 1991.
\newblock Isoperimetry and processes.

\bibitem[M{\etalchar{+}}89]{mcdiarmid1989method}
Colin McDiarmid et~al.
\newblock On the method of bounded differences.
\newblock {\em Surveys in combinatorics}, 141(1):148--188, 1989.

\bibitem[Mag10]{MagdonIsmail10}
Malik Magdon{-}Ismail.
\newblock Row sampling for matrix algorithms via a non-commutative bernstein
  bound.
\newblock {\em CoRR}, abs/1008.0587, 2010.

\bibitem[NN13]{NelsonN13}
Jelani Nelson and Huy~L. Nguyen.
\newblock {OSNAP:} faster numerical linear algebra algorithms via sparser
  subspace embeddings.
\newblock In {\em 54th Annual {IEEE} Symposium on Foundations of Computer
  Science, {FOCS}}, pages 117--126, 2013.

\bibitem[PTZ12]{PengTZ12}
Richard Peng, Kanat Tangwongsan, and Peng Zhang.
\newblock Faster and simpl semidefinite programming.
\newblock In {\em 24th {ACM} Symposium on Parallelism in Algorithms and
  Architectures, {SPAA}}, pages 101--108, 2012.

\bibitem[Sin16]{Singer16}
Yaron Singer.
\newblock Lecture notes.
\newblock
  \url{http://people.seas.harvard.edu/~yaron/AM221-S16/lecture_notes/AM221_lecture9.pdf},
  2016.

\bibitem[Tro08]{tropp2008norms}
Joel~A Tropp.
\newblock Norms of random submatrices and sparse approximation.
\newblock {\em Comptes Rendus Mathematique}, 346(23-24):1271--1274, 2008.

\bibitem[Tro12]{Tropp12}
Joel~A. Tropp.
\newblock User-friendly tail bounds for sums of random matrices.
\newblock {\em Found. Comput. Math.}, 12(4):389--434, 2012.

\bibitem[Ver18]{vershynin2018high}
Roman Vershynin.
\newblock {\em High-dimensional probability: An introduction with applications
  in data science}, volume~47.
\newblock Cambridge university press, 2018.

\bibitem[Woo14]{Woodruff14}
David~P. Woodruff.
\newblock Sketching as a tool for numerical linear algebra.
\newblock {\em Found. Trends Theor. Comput. Sci.}, 10(1-2):1--157, 2014.

\end{thebibliography}

\end{document}